\documentclass[screen]{acmart}

\usepackage[normalem]{ulem}
\useunder{\uline}{\ul}{}
\usepackage{multirow}
\usepackage{enumitem}
\usepackage{graphicx}
\usepackage{amsfonts}
\usepackage[breakable]{tcolorbox}
\usepackage{bm}
\usepackage{booktabs} 
\usepackage{xcolor}
\usepackage[linesnumbered,ruled,vlined]{algorithm2e} 
\tcbuselibrary{breakable, skins}

\newtheorem{theorem}{Theorem}

\newtcolorbox[auto counter, number within=section, use counter=infoBox]{mybox}[2][]{%
    breakable,
    upperbox=visible,
    title=Prompt~\thetcbcounter: #2,
    label=#1
}

\definecolor{orange}{HTML}{EC834A}
\definecolor{green}{HTML}{A2C27F}
\definecolor{SkyBlue}{RGB}{135, 206, 235}  
\definecolor{Lavender}{RGB}{230, 230, 250}  
\definecolor{Peach}{RGB}{255, 218, 185}  
\definecolor{SlateGray}{RGB}{112, 128, 144}  

\AtBeginDocument{%
  \providecommand\BibTeX{{%
    \normalfont B\kern-0.5em{\scshape i\kern-0.25em b}\kern-0.8em\TeX}}}

\setcopyright{acmlicensed}
\copyrightyear{2024}
\acmYear{2024}
\acmDOI{XXXXXXX.XXXXXXX}

\acmISBN{978-1-4503-XXXX-X/18/06}

\begin{document}

\title{Explainable Recommendation with Simulated Human Feedback}

\author{Jiakai Tang}
\email{tangjiakai5704@ruc.edu.cn}
\affiliation{%
  \institution{Gaoling School of Artificial Intelligence, Renmin University of China}
  \city{Beijing}
  \country{China}
}

\author{Jingsen Zhang}
\email{zhangjingsen@ruc.edu.cn}
\affiliation{%
  \institution{Gaoling School of Artificial Intelligence, Renmin University of China}
  \city{Beijing}
  \country{China}
}

\author{Zihang Tian}
\email{tzh2003@ruc.edu.cn}
\affiliation{%
  \institution{Gaoling School of Artificial Intelligence, Renmin University of China}
  \city{Beijing}
  \country{China}
}

\author{Xueyang Feng}
\email{xueyangfeng@ruc.edu.cn}
\affiliation{%
  \institution{Gaoling School of Artificial Intelligence, Renmin University of China}
  \city{Beijing}
  \country{China}
}

\author{Lei Wang}
\email{wanglei154@ruc.edu.cn}
\affiliation{%
  \institution{Gaoling School of Artificial Intelligence, Renmin University of China}
  \city{Beijing}
  \country{China}
}

\author{Xu Chen}
\authornote{Corresponding Author.}
\email{xu.chen@ruc.edu.cn}
\affiliation{%
  \institution{Gaoling School of Artificial Intelligence, Renmin University of China}
  \city{Beijing}
  \country{China}
}

\renewcommand{\shortauthors}{Jiakai Tang et al.}

\begin{abstract}
Recent advancements in explainable recommendation have greatly bolstered user experience by elucidating the decision-making rationale. However, the existing methods actually fail to provide effective feedback signals for potentially better or worse generated explanations due to their reliance on traditional supervised learning paradigms in sparse interaction data.
To address these issues, we propose a novel human-like feedback-driven optimization framework. This framework employs a dynamic interactive optimization mechanism for achieving human-centered explainable requirements without incurring high labor costs. Specifically, we propose to utilize large language models (LLMs) as human simulators to predict human-like feedback for guiding the learning process. To enable the LLMs to deeply understand the task essence and meet user's diverse personalized requirements, we introduce a human-induced customized reward scoring method, which helps stimulate the language understanding and logical reasoning capabilities of LLMs. Furthermore, considering the potential conflicts between different perspectives of explanation quality, we introduce a principled Pareto optimization that transforms the multi-perspective quality enhancement task into a multi-objective optimization problem for improving explanation performance. At last, to achieve efficient model training, we design an off-policy optimization pipeline. By incorporating a replay buffer and addressing the data distribution biases, we can effectively improve data utilization and enhance model generality. Extensive experiments on four datasets demonstrate the superiority of our approach. 
\end{abstract}

\ccsdesc[500]{Information systems~Recommender systems}

\keywords{Explainable Recommendation, Large Language Model, Reinforcement Learning}

\received{20 February 2007}
\received[revised]{12 March 2009}
\received[accepted]{5 June 2009}

\maketitle

\section{Introduction}\label{sec:intro}
In recent years, Explainable Recommender Systems (ERS) have garnered growing attention in both of the academic and industry sectors~\cite{Zhang_2020,10.1145/3331184.3331254}. The fundamental goal of explainable algorithms is to elucidate the rationale behind the recommendation behavior, offering users clear insights into why certain items are recommended. By providing transport explanations, ERS not only facilitate quicker decision-making for customers but also significantly better overall satisfaction. This paradigm shift towards trustworthy recommendation is driven by the increasing demand for deeper user engagement in automated applications. 
Powered by potent language models, ERS transform the explainable recommendation task into the user review generation problem. In specific, Attribute-to-Sequence (Att2Seq)~\cite{dong2017learning} and Neural Rating and Tips generation (NRT)~\cite{li2017neural} leverage LSTM~\cite{hochreiter1997long} and GRU~\cite{69e088c8129341ac89810907fe6b1bfe} based RNN architectures respectively to convert latent user and item representations into sentence outputs, aiming to maximize the likelihood of the target texts. Due to insufficient guidance on informative content generation, these approaches often suffer from uncontrolled explanation quality. To address this issue, recent work, such as NETE~\cite{li2020generate}, PETER~\cite{li-etal-2021-personalized}, ERRA~\cite{DBLP:conf/acl/ChengWLZ0LL23}, and PEPLER~\cite{li2023personalized} have been developed, they utilize up-to-date architectures (\emph{e.g.}, Transformer~\cite{vaswani2017attention}, GPT-2~\cite{radford2019language}, etc.) as backbones and integrate additional augmented feature (\emph{e.g.}, pre-extracted user preference aspects like ``screen$\to$small'') for promoting the personalized explanation generation.

While recent advancements in the filed of ERS have demonstrated promising achievements, they also face some severe drawbacks that cannot be ignored. As depicted in Figure~\ref{fig:introduction}, \textbf{firstly}, in real-world recommendation scenarios, 
high-quality recommendation explanation data are often scarce. The insufficient high-quality training data impairs the model's ability to accurately capture genuine user preferences and item characteristics, particularly affecting the performance for inactive users and niche items. Consequently, these issues of data sparsity ultimately lead to suboptimal model robustness, that is, reducing the recommendation and explanation quality for disadvantaged groups. \textbf{Secondly}, personalized recommendation explanations that better align with users' subjective preferences can lead to improved user satisfaction. Therefore, the text quality should ideally be assessed by humans instead of relying on various text similarity metrics as a compromise.
However, traditional methods predominantly follow a supervised learning paradigm, focusing on maximizing the likelihood of generated target texts. Such approach fails to distinguish between potentially superior or inferior generated explanations, thereby impeding the generalization performance in complicated application scenarios. For example, in the right half of Figure~\ref{fig:introduction}, we can easily see that the explanation ``My skin feels softer and smooth as it absorbs quickly'' is more informative and persuasive than ``It feels nice on my skin'' for user $u_2$ when purchasing skincare product $v_2$. However,
due to blindly focus on fitting the target texts in the corpus, it cannot effectively provide the correct optimization signal for model outputs that are potentially better (or poorer) than the ground truths. This constraint in turn restricts the overall performance of model's generated sentences.

\begin{figure}[t]
\centering
\includegraphics[width=\textwidth]{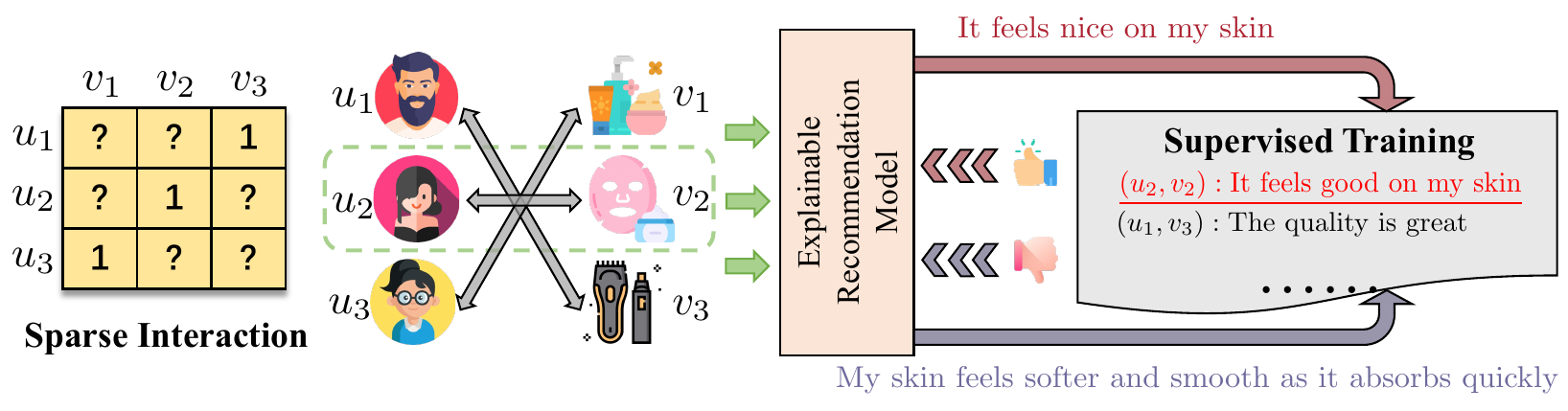}
\caption{Motivation Example. The left half illustrates the longstanding issue of data sparsity in ERS, where most user-item interactions are unobserved or unavailable. The right half demonstrates the supervised training paradigm employed by traditional explainable recommendation models, which blindly maximizes the likelihood of target sentences in the given corpus, despite the potential quality of model-generated texts to be better than the ground truths.} 
\label{fig:introduction} 
\end{figure}

In response to the limitations, we innovatively propose to reframe the optimization process within a reinforcement learning framework, integrating human feedback to dynamically guide the model's training process. This paradigm shift from static textual fitting to interactive optimization, which provides effective signals for unobserved user-item pairs and generate sentences more aligned with human tastes. Although this seems to be an intriguing and novel idea, there are actually many challenges:
\begin{enumerate}[label=C\arabic*:, leftmargin=*]
    \item \textbf{Prohibitive Costs of Human Involvement.} Direct human participation in obtaining real-time feedback during the training process is clearly impractical and inaccessible. The expensive costs in terms of both time and resources make our initial idea unfeasible for large-scale corpus.
    \item \textbf{Complicated Reward Mechanism Design.} Due to the heterogeneous natures of different user preferences and item characteristics, establishing a uniform and fixed reward scoring mechanism may struggle to deeply understand measure criteria of explanation quality under the reinforcement learning optimization framework. Furthermore, how to reasonably incorporate valuable context information within the reward function is also important for effective and personalized rewards.
    \item \textbf{Multi-Perspective Explanation Optimization.} The evaluation criteria of the explanation quality is usually multifaceted from the human view, encompassing different perspectives such as persuasiveness and informativeness. Furthermore, these measure metrics are even mutually contradictory to some extent. For example, a sentence like ``This skincare product is hot!'' might be relatively persuasive but lacks meaningful information. Conversely, some explanations, such as `This is a 2 fl oz bottle of our Acne Stop Skin Care, approximately one months supply.', include detailed product information, yet it does not inspire any desire to click and purchase among consumers. Thus, how to design multi-perspective optimization mechanism also challenges our idea.
\end{enumerate}

 To address the above challenges, we propose a novel optimization framework for explainable recommendation, named as \textbf{HF4Rec} for short. The basic idea of HF4Rec is to leverage the powerful anthropomorphic capability of large language models (LLMs) as human simulators, which drive explanation enhancements through human-like information feedback. Specifically, we firstly introduce human preference-based prompt prototypes to automatically elicit customized reward criteria for different user-item interactions, aiming to refine reasoning logic for the personalized scoring mechanism of explainable recommendation task. Then, to overcome the challenges posed by lengthy and noisy user history interactions on the large model's ability to understand user interests, we propose target-aware retrieval strategy to extract key information from history behaviors. At last, we design the principle-guided Pareto optimization framework to improve the multi-perspective explanation quality in a fully differentiable manner. This approach effectively overcome the potential conflicts between different evaluation metrics, such that HF4Rec can simultaneously maximize all the objective functions. In the experiment, we apply our framework to different base models on four real-world public datasets to demonstrate the effectiveness of the proposed HF4Rec.

In a summary, the main contributions of this paper are as follows:
\begin{itemize}[leftmargin=*]
    \item We innovatively reframe the traditional supervised text-fitting paradigm into human-like feedback-driven reinforcement learning optimization method, emphasizing the essential human-centered requirements for explainable recommendation task.
    \item We introduce large language models as human simulators and meticulously design personalized explanation scoring mechanism aligned with human preference as reward functions. This further explores more augmented unobserved user-item interactions, achieving robust explainable recommendation performance.
    \item We propose a principled Pareto optimization framework to enhance the explanation quality across multiple perspectives, effectively mitigating potential conflicts among different evaluative aspects.
    \item We conduct extensive experiments based on different base models on four real-world public datasets, consistently outperforming baselines in terms of objective textual similarity metrics and subjective human-view multi-perspective explanation quality. 
\end{itemize}

In what follows, we organize the rest of our paper as follows: Section 2 introduces preliminary knowledge. Section 3 presents our proposed Human-Like Feedback-Driven optimization framework, and Section 4 presents the experimental results and analysis. Section 5 discusses the related work of traditional and LLM-augmented explainable recommendation, followed by final conclusions and future work in Section 6.

\section{Preliminary}
\subsection{Notation and Task Formulation}
Before delving into the details of our proposed method, we first introduce basic notations and formally define the explainable recommendation task. In real-life scenarios, obtaining genuine explanations for user recommendations are quite difficult. Following existing work~\cite{zhang2023recommendation,li2023personalized,li-etal-2021-personalized}, we extract recommendation explanations from user reviews, ensuring that each explanation contains at least one item feature (e.g., taste) to ensure explanation quality. Formally, we define a user set $\mathcal{U}$ and item set $\mathcal{I}$. For observed user-item pairs $\mathcal{D}=\{(u,v)|u\in \mathcal{U}, v\in \mathcal{I}\}$, each data example $(u,v)$ is associated with $r_{u,v}$ and $\mathbf{x}_{u,v}$, where $r_{u,v}$ is the user rating within the range of [1,5] and $\mathbf{x}_{u,v}$ denotes the review commented by user $u$ on item $v$. The review $\mathbf{x}_{u,v}$ consists of a sentence with length $l_{u,v}$, \emph{i.e.}, $\mathbf{x}_{u,v}=\{x^1_{u,v},x^2_{u,v},\dots,x^{l_{u,v}}_{u,v}\}$, and we define the word vocabulary as $\mathcal{V}$. In addition, we denote the history items interacted with user $u$ as $\mathbf{h}_u$, and each item is characterized by several feature information, such as category, product description, title, etc. We provide the key notations used throughout the paper in Table~\ref{tab:notations}.

\begin{table}[t]
    \centering
    \caption{Notations}
    \begin{tabular}{cl}
        \hline\hline
        \textbf{Symbol} & \multicolumn{1}{c}{\textbf{Description}}\\
        \toprule 
        $\mathcal{U}, \mathcal{I}$ & the user set and the item set \\
        $\mathcal{D}, \tilde{\mathcal{D}}$ & the observed and unobserved interaction data \\
        $r_{u,v}, \hat{r}_{u,v}$ & the ground truth and predicted rating for $(u,v)$ pair \\
        $\mathbf{x}_{u,v}, \hat{\mathbf{x}}_{u,v}$ & the ground truth and predicted explanation for $(u,v)$ pair\\
        $\mathcal{V}$ & the word vocabulary\\
        $\mathbf{h}_u$ & the history interaction of user $u$ \\
        $l_{u,v}$ & the sentence length of user $u$'s review to item $v$\\
        $\psi$ & the reward value \\
        $A$ & the advantage value \\
        $N_u^{C_i}$ & the number of interactions user $u$ has with items of category $C_i$\\
        $p_u^{C_i}$ & the sampling probability of user $u$ in item category $C_i$ \\
        $Z$ & the total number of item categories \\
        $J$ & the number of explanations explored for each user-item pair\\
        \hline\hline
    \end{tabular}
    \label{tab:notations}
\end{table}

\noindent
\textbf{Task Formulation of Explainable Recommendation.} In essence, explainable recommendation methods leverage multi-task learning framework. More specifically, given the user and item identifications, the objectives are twofold as follows: (1) It is required to accurately predict the rating $\hat{r}_{u,v}$ that represents user $u$'s preference towards item $v$. (2) The well-trained model should concurrently generate a textual explanation $\hat{\mathbf{x}}_{u,v}$, which serves to illuminate the rational behind the recommendation.

\subsection{Traditional Explainable Recommendation}
In traditional explainable recommendation models (ERMs), the model optimization is conducted based on the supervised learning paradigm using observed training data. Here, we outline the process in two main phases as follows:
\begin{itemize}[leftmargin=*]
    \item \textbf{Training Phase.} The model employs a multi-task learning mechanism to simultaneously perform rating prediction task and explanation generation task. Specifically, for the rating prediction task, ERMs use the Mean Square Error (MSE) as the regression loss function:
    \begin{equation}
        \mathcal{L}^R = \frac{1}{|\mathcal{D}|}\sum_{(u,v)\in \mathcal{D}} (r_{u,v}-\hat{r}_{u,v})^2,
    \end{equation}
    where $r_{u,v}$ and $\hat{r}_{u,v}$ represent the ground truth and predicted rating for the given $(u,v)$ pair, and $|\mathcal{D}|$ denotes the size of observed training data.
    For the explanation generation task, we adopt the Negative Log Likelihood (NLL) as the loss function to compute the average loss of the explanations generated for the user-item pairs observed in the training corpus, which is formulated as follows:
    \begin{equation}
        \mathcal{L}^E = \frac{1}{|\mathcal{D}|}\sum_{(u,v)\in \mathcal{D}} \frac{1}{l_{u,v}}\sum_{t=1}^{l_{u,v}}-\log h(\hat{x}_{u,v}^t=x_{u,v}^t | u,v,\mathbf{x}_{u,v}^{1:t-1}),
    \end{equation}
    where $h(\cdot)$ can be implemented with any off-the-shelf sequential architecture like GRU, LSTM and Transformer. ``$\hat{x}_{u,v}^t=x_{u,v}^t$'' denotes the probability that the ERM predicts the next token at position $t$ as $x_{u,v}^t$ in the ground truth explanation. Given the user, the item, and the preceding tokens $\mathbf{x}_{u,v}^{1:t-1}$, traditional methods follow such teacher-forcing mechanism to learn explainable knowledge. In fact, the cross-entropy based learning objective aims to maximize the log-likelihood of the target sentence, overlooking the real quality of the generated explanations if they do not strictly align with the predefined tokens, which is the central problem that our paper focuses on.
    \item \textbf{Inference Phase.} Given the user and item identifies, the ERM starts the generation process with the begin-of-sentence token <bos>. It greedily outputs the next token that has the highest probability, then merges this token with the preceding sequence to form the new input. This iterative process continues to generating subsequent words until the sentence either contains with the end-of-sequence token <eos> or reaches the predefined maximum sequence length. While for the rating prediction task, different explainable models have their unique implementation strategies. For instance, in the PETER~\cite{li-etal-2021-personalized}, the scalar value obtained from the Feed-Forward Network (FFN) applied to the output at the first position of the last layer represents the final predicted rating. 
\end{itemize}

\subsection{Policy-based Reinforcement Learning}
In typical reinforcement learning framework, the overall environment is modeled using the following quadruple: (1) \textbf{State Space ($\mathcal{S}$)}: All possible states for guiding the agent's decision-making process. (2) \textbf{Action Space ($\mathcal{A}$)}: All possible actions that the agent can take. (3) \textbf{Transition Probability ($\mathcal{P})$}: The probability that the agent moves to state $s'$ from state $s$ after taking action $a$, denoted as $p(s'|s,a)$. (4) \textbf{Reward Function ($\mathcal{R}$)}: The reward the agent receives for transitioning from state $s$ to state $s'$ after taking action $a$. It defines the Markov Decision Process (MDP) in reinforcement learning, with the objective to develop a policy $\pi(a|s)$ that maximizes the expected total return. In general, Reinforcement Learning (RL) techniques can be categorized into value-based~\cite{watkins1992q,van2016deep,mnih2015human} and policy based~\cite{schulman2015trust,konda1999actor,DBLP:journals/corr/LillicrapHPHETS15} approaches. Specifically, the former computes a Q-value for each state-action pair, facilitating the agent to choose action with the highest reward. On the other hand, the policy-based methods directly seek an optimal policy by parameterizing the policy itself, which is more advantageous in environments with large action spaces.

Given the expansive action space in the explainable recommendation task (in specific, containing 50260-sized token candidates as the action space in our experiments) and grounding on empirical performance considerations, we adopt the policy-based reinforcement learning approach. Formally, the parameter updates in policy-based models can be expressed as follows:
\begin{equation}
    \label{naive_pg}
    \theta_{t+1} = \theta_t + \gamma \nabla_{\theta} \log \pi_\theta (a_t | s_t) \psi_t,
\end{equation}
where $\theta$ denotes the parameters of the policy network, $\gamma$ is the learning step, $\nabla_{\theta} \log \pi_\theta (a_t | s_t)$ is the gradient tensor of the logarithm of the policy probability $w.r.t.$ the parameters $\theta$, and $\psi_t$ is the reward received after taking action $a_t$. The optimization goal is to adjust the policy network in directions that could maximize the expected reward.

\section{METHODOLOGY}
In this section, we present the proposed approach, a Human-Like Feedback-Driven Optimization Framework for Explainable Recommendation (named as HF4Rec for short). The overview of HF4Rec is demonstrated in Figure~\ref{fig:framework}. The general idea of HF4Rec is to utilize large language models (LLMs) as human simulators within a reinforcement learning framework for achieving dynamic interactive optimization mechanism. By adopting such human-in-loop method, we can effectively improve the multi-perspective quality of recommendation explanations. Specifically, by reframing the explainable recommendation into a policy-based reinforcement learning approach, HF4Rec employs an off-policy ``data collection and model updating'' loop for efficient model parameter optimization. In the phase of data collection, to overcome the challenges such as data sparsity and the heterogeneity of user-item interactions, we propose the difficulty-aware sampling strategy and retrieval-augmented customized reward prompt. In the second phase, we adopt principled Pareto optimization strategy to improve multi-perspective quality of recommendation explanations, effectively mitigating the potential conflicts between different evaluation criteria.

In the following, we elaborate on the methodology design of the proposed HF4Rec. 

\begin{figure}[t]
\centering
\includegraphics[width=\textwidth]{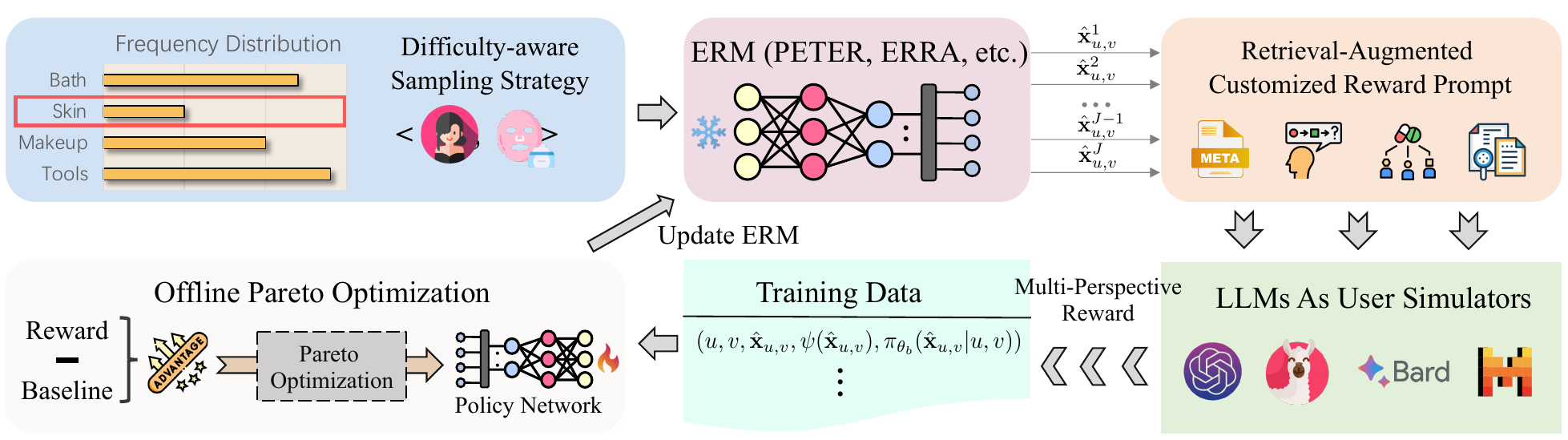}
\caption{The Overview of the proposed HF4Rec.} 
\label{fig:framework} 
\end{figure}

\subsection{Off-Policy RL Framework for Explainable Recommendation}
Traditional explainable recommendation algorithms mostly rely on supervised learning techniques, \emph{i.e.}, focusing on maximizing the likelihood of observed user reviews in the corpus. However, the paradigm of text-fitting essentially limits the model's ability to explore and evaluate unobserved but potentially valuable generated explanations. To overcome this constraint, we introduce a reinforcement learning framework that leverages human feedback to update policy network, which allows the model to update its parameters based on the human-judged explanation quality.
Inspired by the emergence of large language models demonstrating unparalleled abilities in language understanding~\cite{touvron2023llama,DBLP:conf/iclr/ZengLDWL0YXZXTM23,brown2020language,huang2022towards,DBLP:conf/iclr/ZhouSHWS0SCBLC23} and their exceptional performance in role-playing humans across various tasks~\cite{yoon2024evaluating,wang2023user,aher2023using}, we can utilize available LLMs as human simulators to attain the human-in-the-loop iterative process, avoiding the resource and time costs of real-human involvement. 
In the next, we will introduce the proposed off-policy reinforcement learning framework, which primarily addresses the following critical issues: (1) How can we effectively explore unobserved data and generate diverse recommendation explanations to enhance the model robustness and generalization ability? (2) How can human-like feedback information be efficiently utilized to optimize model parameters? (3) How can HF4Rec achieve the stable model learning process under policy-based RL optimization framework?

\subsubsection{\textbf{Difficulty-aware Sampling Strategy}}\label{sec:difficulty_aware_sampling}
To tackle the decreased model robustness due to the data sparsity and popularity bias, we propose a difficulty-aware sampling strategy, which serves for the subsequent data collection phase. Inspired by the work~\cite{ovaisi2022rgrecsys}, an ideal robust recommendation algorithm should perform well across all subgroups, especially within groups characterized by limited interactions. Therefore, our general idea is to boost the sampling chances for item categories with fewer interactions while decreasing them for those with higher record numbers. Formally, suppose that $(N_u^{C_1}, N_u^{C_2}, \dots, N_u^{C_Z})$ represent the interaction frequency for user $u$ with items across categories $(C_1, C_2, \dots, C_Z)$, where $Z$ is the total number of item categories. To enhance average performance in under-represented categories, we define the sampling probability of the item category $C_i$:
\begin{equation}
    \label{eq:prob_sampling}
    p_u^{C_i} = \frac{1/\log(N_u^{C_i}+2)}{\sum_{z=1}^Z 1/\log(N_u^{C_j}+2)},
\end{equation}
where the logarithm function is used to prevent extreme frequency distributions from causing a minority of item categories to dominant the sampling process, and the constant is to avoid numerical instability. While for items within the same category, we adopt the random sampling method to simplify the overall complexity, and leaving further complex methods for the future work.

\subsubsection{\textbf{Off-Policy Optimization Pipeline}}
Upon determining the user-item training samples, we can feed the data into the human simulators, which yields reward estimations for the model's generated explanations. Then, these rewards are used to update the model parameters using the policy gradient method as specified in Eq.~(\ref{naive_pg}). However, this method actually faces two weaknesses as follows: (1) \textbf{Serial Inference Latency}. Since both the behavior policy and target policy are the same one, and input data continuously interact with the large language model, there is a serial buildup of inference delays, greatly increasing the overall training time. (2) \textbf{Risk of Converging to Local Optimal}. Traditional online policy gradient methods prioritize interactions under the current policy, overlooking the valuable knowledge available from past policies. This inadequate data utilization could steer the model towards a local suboptimal position, causing the model miss opportunities to develop better policy parameters.

To address the above issues, we propose the following two-phase pipeline:
\begin{enumerate}[leftmargin=*]
    \item \textbf{Data Collection Phase}. During this phase, we initially derive all training data from observed interactions $\mathcal{D}$ and unobserved user-item interactions $\tilde{\mathcal{D}}$ obtained with the proposed data sampling strategy(detailed in Sec~\ref{sec:difficulty_aware_sampling}). After that, user-item pairs are fed into the pretrained explainable recommendation model to generate explanation predictions, which are then formatted using predefined prompt templates (detailed in Sec~{\color{red}{\ref{sec:prompt}}}), and the human simulator (\emph{i.e.}, LLM) receives these input texts and assigns rewards based on the explanation quality like humans. Subsequently, all trajectories are stored in a replay buffer $\mathcal{B}$ for the next phase (\emph{i.e.}, model updating), where the basic structure of each data entry consists of $(u,v, \hat{\mathbf{x}}_{u,v},\psi({\hat{\mathbf{x}}_{u,v}}),\pi_{\theta_{b}}(\hat{\mathbf{x}}_{u,v}|u,v))$. It is worth noting that it does not involve updating model parameters in this phase, thus we can parallelly carry out the entire data collection process on a large scale, dramatically speeding up the overall pipeline.
    \item \textbf{Model Updating Phase}. In this phase, the trajectory data collected in the last phase are utilized to update the policy network, \emph{i.e.}, the explainable recommendation model. However, Different from naive online policy gradient methods, the behavior policy that interacts with the environment is not the same as the target policy to be optimized, and does not update concurrently with the target policy. This discrepancy leads to an undesired inconsistency in the example distribution between two policies. To address this data bias, we utilize a surrogate objective function that corrects data distribution by incorporating important sampling weights. In addition, inspired by Proximal Policy Optimization (PPO)~\cite{schulman2017proximal}, we use clipping technique to discard samples that fall outside the trust region, preventing large deviations in policy optimization. Specifically, we redefine the objective function as follows:
    \begin{equation}
        \label{eq:obj}
        \underset{\theta}{\text{arg} \max}\ \mathbb{E}_{\pi_{\theta_{b}}} \left[ 
        \min\left( \frac{\pi_\theta(\hat{\mathbf{x}}_{u,v}|u,v)}{\pi_{\theta_{b}} (\hat{\mathbf{x}}_{u,v}|u,v)} \psi(\hat{\mathbf{x}}_{u,v}), \text{clip}\left(\frac{\pi_\theta(\hat{\mathbf{x}}_{u,v}|u,v)}{\pi_{\theta_{b}}(\hat{\mathbf{x}}_{u,v}|u,v)} ,1-\epsilon, 1+\epsilon\right)
        \psi({\hat{\mathbf{x}}_{u,v}})\right)\right],
    \end{equation}
    where $\theta$ and $\theta_b$ represent the network parameters of the target policy and behavior policy, respectively, and the hyper-parameter $\epsilon$ stands for the clipping range. All reward values correspond to the actions generated by the behavior policy $\pi_{\theta_b}$.
\end{enumerate}

Moreover, we can also adopt multiple iterations of data collection and model updating, mitigating the negative effects caused by excessive distribution gap between the behavior policy and the target policy. Note that our optimization pipeline differs from on-policy methods as we leverage the replay buffer to recycle gathered trajectories across multi-turn iterations, thereby enriching the model’s knowledge base and boosting the generalization ability.

\subsubsection{\textbf{Advantage Estimation for Stable Model Learning}}\label{sec:adv}
The objective function based on reward values may often lead to learning instability because of the high estimation variance and noise~\cite{mnih2016asynchronous,10.5555/3042573.3042600}. To address this issue, we propose an advantage function for more stable model training process. Specifically, utilizing the probability-based token prediction method within an explanation recommendation model, we can explore $J$ different explanations for the same user-item input pair. Then, the average value of the rewards can be viewed as approximated estimation of the state value, which allows us to derive the advantage value as follows:
\begin{equation}
    \label{eq:adv_value}
    A(\hat{\mathbf{x}}_{u,v}^i) = \psi({\hat{\mathbf{x}}_{u,v}^i}) - \frac{1}{J} \sum_{j=1}^J \psi({\hat{\mathbf{x}}_{u,v}^j}),
\end{equation}
where $\hat{\mathbf{x}}_{u,v}^i$ denotes $i$-th generated explanation of $(u,v)$ pair.
Compared with reward values in Eq~(\ref{eq:obj}), advantage values can better guide the policy network to more rapidly differentiate the different quality of generated sentences. In particular, $A(\cdot)>0$ means that the text quality of the $i$-th action is above the average level, whereas $A(\cdot)<0$ signifies the below-average quality. This mechanism encourages the ERM to update parameters in the positive direction for the former and in the opposite direction for the latter. Furthermore, incorporating advantage function into the policy gradient optimization can reduce the high variance, thereby upgrading the model's learning stability.

\subsection{Retrieval-Augmented Customized Reward Prompt}\label{sec:prompt}
To effectively harness large language models (LLMs) for predicting human-like reward scores, we have to overcome the following two challenges in prompt design. 
\textbf{Firstly}, user interaction histories often include noisy data (such as friend recommendation, proxy purchasing, clickbait, etc.). Thus this could potentially dilute the true user preferences and interests, impeding the LLM's ability to precisely extract relevant information from long-context input. \textbf{Secondly}, since the behaviors between different users and items exhibit significant pattern heterogeneity, establishing a uniform and fixed criteria for reward measure may struggle to encompass the complexity and diversity of multi-perspective explanations. Therefore, aligning human simulators with the personalized human standards also poses a tricky challenge. To this end, we propose a novel retrieval-augmented customized reward prompt, which mainly consists of the two steps: (1) retrieval-augmented user information extraction. (2) customized human-induced scoring mechanism.  

\subsubsection{\textbf{Retrieval-Augmented User Information Extraction}}
The first step to interact with the human-like simulators is reformulating the user and item information in a structured description conducive to large language model processing. To achieve this, we transform related features into prepared hard coded prompt template using ``\{key: value\}'' pairs. This format organizes important features such as item category, title, description, etc., in a dictionary format that is easily parsed and understand by LLMs. And for items in user behavior history, we add user review as extra item content description. Then we give a specific transformation example as shown in Prompt~\ref{prompt:user_item_info}: 

\begin{mybox}[prompt:user_item_info]{The Description of Item Information}
\{\{

\quad ``Item Title'': `AINHOA Luxe Hydra Luxe Absolute Gift Set (Body Cream, Anti-Aging Cream)'\ ,

\quad ``Item Description'': `AINHOA is an award-winning series of professional beauty products for use in the comfort of your home.  Pleasurable, convenient, made from the highest quality ingredients, these are Spa-quality products at affordable prices ...'\ ,

\quad ``Item Category'': `Skin Care'\ ,

\quad [Optional Feature] ``Recommendation Explanation'': `While the cream felt great going on and soaked into my skin quickly.'

\}\}
\end{mybox}

To extract valuable information from the lengthy and noisy user interaction history, we introduce a target-aware retrieval-augmented method. Specifically, we leverage the OpenAI embedding model\footnote{text-embedding-3-small:\url{https://platform.openai.com/docs/api-reference/embeddings}} to acquire sentence embeddings for all interacted items and the target item. By doing this, we can access the semantic relevance between each candidate behavior and the target item by calculating the cosine similarity between the corresponding  semantic embeddings. Afterwards, we can identify and select the top-$K$ item with the highest similarity scores to stand for target-aware user description, which effectively reflects the user's interest preference.

\subsubsection{\textbf{Customized Human-Induced Scoring Mechanism}}
Ideally, to define personalized reward scoring mechanisms that align with different human tastes to explanation quality, we could craft human-written examples for each user-item pair to cover different reward patterns in a few-shot manner. However, such an approach would be prohibitively costly due to the extensive manual effort required. In real-world scenarios, user presences are usually diverse within a high-level feature space, and interaction history sequences from different users can share similarity in such feature space. 
Inspired by this idea, we can initially group users into $Y$ clusters based on their extracted information description using any clustering algorithms such as $K$-means~\cite{hartigan1979algorithm}. This step helps to find distinct user groups with similar styles. Subsequently, we randomly select one representative user from every cluster and manually create few-shot scoring examples as prompt prototypes.

Following that, in order to further refine the instance-specific customized knowledge for achieving more comprehensive and personalized understanding of each user-item pair, we utilize the prompt prototype corresponding to the user's cluster, and then induce the large language model to create a new customized reward few-shot context. To better illustrate the detailed implementation, we present an example as demonstrated in Prompt~\ref{prompt:context_prompt}, where \textcolor{orange}{words in orange} signifies placeholders for the retrieved user in information, item description, and few-shot examples.

\begin{mybox}[prompt:context_prompt]{Human-Induced Few-shot Context Generation}
Given the following examples:


Here is information on items previously interacted with by a user, along with their corresponding recommendation explanations:

\textcolor{orange}{\{User interaction history with top-$K$ most similar item information\}}

Then, information about another item is given as follows:

\textcolor{orange}{\{Target item information\}}

We present several user's recommendation explanations about this item, along with two scores (1-3) based on

a. Informativeness: Help user learn more about the recommended item;

b. Persuasiveness: Help user want to buy the item.

{\color{orange}
[

\quad\{\{
    
\quad\quad``Recommendation Explanation'': ``The scent is the pleasant light scent of the other Clear shampoos.'',
        
\quad\quad``Informativeness'': 3,
        
\quad\quad``Persuasiveness'': 3,
        
\quad\quad``Reason'': ``Informativeness (3): It directly informs information about the product's benefits. 

\quad\quad Persuasiveness (3): The description of the scent as `pleasant light' is positive and appealing, making the product desirable for those who prefer subtle, agreeable fragrances in their hair care products.''
        
\quad \}\},

\quad \textcolor{green}{\# More examples to cover other scores}

]}


A new user has interacted with the following items:

\textcolor{orange}{\{User interaction history with top-$K$ most similar item information\}}

And we have information about the target item:

\textcolor{orange}{\{Target item information\}}

Based on the above examples, generate three explanation examples and corresponding reasons for this user about the target item. 

**IMPORTANT NOTE**

1. Ensure output MUST adhere to the following format:

\{\{

\quad``Explanation'': \# less than 15 words,
    
\quad``Informativeness'': \# [1-3],
    
\quad``Persuasiveness'': \# [1-3],
    
\quad``Reason'': \# less than 100 words
    
\}\}.

2. The generated examples must cover all score levels, that is, 1, 2, and 3.

3. Please generate results directly, without any additional thinking process.

\end{mybox}

Furthermore, as illustrated in Prompt~\ref{prompt:context_prompt}, we develop a Three-Tier Instruction Enhancement strategy for fully exploiting the chain of thought potentiality in large language model. Specifically, to begin with, we explicitly define the quality standards as a transport reward guideline with respect to informativeness and persuasiveness of generated explanations. Secondly, customized few-shot contexts equip LLM with refined cases across various scores, facilitating the implicit understanding of intricate reward mechanism. At last, we introduce in-depth reasons where quality scores of the explanation example are justified in detail for aligning with user preference. This better assists the LLM to learn rationale behind different scores, awakening the LLM's logical reasoning ability to accurately infer reward values.

After deriving the customized context prompts, we proceed to integrate the available information and feed them into the LLM to produce human-like reward outputs. More specifically, the input prompt comprises several key components: (1) Task description; (2) User interest modeling; (3) Target Item Construction; (4) Three-Tier Instruction Enhancement. For provide clearer descriptions, we present the reward prompt template as shown in Prompt~\ref{prompt:reward_prompt}.

\begin{mybox}[prompt:reward_prompt]{Reward Prompt}
You need to act as an explainable recommender system, providing explanations used to help the user understand why an item was recommended.

A user has interacted with the following items:

\textcolor{orange}{\{User interaction history with top-$K$ most similar item information\}}

And we have information about the target item:

\textcolor{orange}{\{Target item information\}}

Here are some explanations about recommending the target item to this user, and scores (from range [1-3]) are given from the perspectives of 

a. Informativeness: Help user learn more about the recommended item;

b. Persuasiveness: Help user want to buy the item.

\#\#\# Example BEGIN \#\#\#

\textcolor{orange}{\{Customized Few-shot Examples\}}

\#\#\# Example END \#\#\#

A new recommendation explanation is provided as follows:

\textcolor{orange}{\{Generated Recommendation Explanation\}}

Learn from the above examples to evaluate new explanation.

**IMPORTANT NOTES**

1. The output format MUST be:

\{\{

\quad ``Informativeness'': \# [1-3],
    
\quad ``Persuasiveness'': \# [1-3]
    
\}\}.

2. Please ensure that the scoring for the two perspectives of explanation quality is done independently, without influencing each other.

\end{mybox}

\subsection{Multi-Perspective Explanation Optimization Based on Pareto Optimality}
The quality assessment of recommendation explanations often involves various perspectives, such as informativeness and persuasiveness. However, as highlighted in Section~\ref{sec:intro}, these evaluation metrics do not always align harmoniously. For example, while an explanation that emphasizes persuasiveness may captivate the user's attention, it could omit critical details regarding the product's unique characteristics or user experience. On the other hand, an explanation that meticulously details product attributes may not always effectively engage the user or motivate purchasing decisions. Therefore, we treat the task of improving multi-perspective quality of recommendation explanations as a multi-objective optimization problem, aiming to strike a balance between different aspects.

Traditional methods for addressing multi-objective task typically require presetting weights for each objective function. However, these methods usually encounter the following three limitations: (1) They suffer from tedious trial-and-error, which consumes expensive time and resource costs. (2) The static weights fail to dynamically adjust during training, which could result in suboptimal performance as different objectives may vary greatly in scale and importance. (3) The simple weighted sum approach does not guarantee balanced improvements across all aspects. 

To overcome the above drawbacks, we propose a multi-perspective explanation optimization method based on Pareto Optimality~\cite{ribeiro2012pareto,censor1977pareto,hochman1969pareto,luc2008pareto,rame2023rewarded}, which strikes a dynamic balance between different objectives under theoretical guidance. In specific, assume there are $M$ objective functions (such as informativeness and persuasiveness) in ERS denoted as $\{\mathcal{F}_1(\theta),\mathcal{F}_2(\theta),\dots,\mathcal{F}_M(\theta)\}$, where $\mathcal{F}_i(\theta)$ represents the $i$-th objective function. In this framework, we can derive the respective objective functions by substituting the rewards in Eq.~(\ref{eq:obj}) with the advantage values corresponding to each quality aspect. And in the next, we formalize several definitions \emph{w.r.t.} Pareto optimization as follows:
\begin{enumerate}[label=\arabic*.]
    \item \textbf{Pareto Domination.} A solution $\theta_1$ dominates another solution $\theta_2$ if and only if $\mathcal{L}_i(\theta_1)\geq \mathcal{L}_i(\theta_2)$ for all objectives, and there exists at least one objective $j$ (where $1\leq j \leq M$) such that $\mathcal{L}_j(\theta_1) > \mathcal{L}_j(\theta_2)$.
    \item \textbf{Pareto Optimality.} A solution $\theta$ is defined as Pareto optimality when it cannot be further improved on any objective without degrading others. In other words, there is no other solution that dominates the solution $\theta$.
    \item \textbf{Pareto Front.} For multi-objective tasks, the Pareto front represents a set of all Pareto optimal solutions.
\end{enumerate}

To balance the conflicts between different objectives, we adopt the scalrization method~\cite{lin2019pareto,sener2018multi,xie2021personalized,chen2021reinforcement} to average $M$ functions using weights $\bm{\omega}=\{\omega_1,\omega_2,\dots,\omega_M\}$, then we can derive the loss function as follows:
\begin{equation}
    \mathcal{L}(\theta) = - \sum_{i=1}^M \omega_i \mathcal{F}_i(\theta).
\end{equation}
Here, the weights $w$ can be determined by solving the following quadratic programming problem:
\begin{equation}\label{eq:weights}
\begin{gathered}
    \min_{\bm{\omega}} || \sum_{i=1}^M \omega_i \nabla_{\theta} \mathcal{F}_i(\theta) ||_2^2\\
    s.t. \quad \mathbf{1}^\top \bm{\omega} = 1, \\
    \omega_i \geq 0, \quad \forall i \in [1,M], \\
    \mathbf{h}_q^\top\bm{\omega} \geq \beta_q, \quad  \forall q \in [1,Q],
\end{gathered}
\end{equation}
where $\mathbf{1}$ is an all-ones vector. The sets $\{\mathbf{h}_1,\mathbf{h}_2,\dots,\mathbf{h}_Q\}$ and $\{\beta_1,\beta_2,\dots,\beta_Q\}$ represent non-negative $M$-dimension prior preference vectors and scalar values, respectively, where the element-wise sum of $\mathbf{h}_q$ equals one. Then, we can easily incorporate our human prior knowledge about the different aspects of expiation quality by setting the corresponding $\mathbf{h}_q$ and $\beta_q$. Specifically, $\mathbf{h}_q$ is defined as a unit vector while $\beta_q$ is set as a predefined bias to quantify the relative importance corresponding $i$-th evaluation perspective.

Different from the previous methods, our scalar weights are dynamically and adaptively updated during the training process without human intervention, making the entire framework end-to-end. Thereafter, we further elucidate the theoretical analysis why the weights of objective functions solved by Eq.~(\ref{eq:weights}) can lead to the Pareto-optimality solution.

\begin{theorem}
The weights $\bm{\omega}$ derived from the solution Eq.~(\ref{eq:weights}) ensure that either one of the following conditions is satisfied:
\begin{enumerate}[label=\roman*.]
    \item The solution to this minimal optimization problem is 0, satisfying the Karush-Kuhn-Tucker (KKT) conditions, which confirms that a local Pareto optimal solution is achieved.
    \item The solution gives a gradient direction capable of improving all objectives simultaneously.
\end{enumerate}
\end{theorem}

\begin{proof}
For condition \romannumeral1., if the minimal solution is 0, it indicates that $\sum_{i=0}^M \omega_i \nabla_\theta \mathcal{F}_i(\theta)=0$, which proves that $\theta$ has reached a Pareto stationary point where no further improvement in all objectives is possible. 

For condition \romannumeral2., we reformulate the quadratic programming in Eq.~(\ref{eq:weights}) into the following Lagrangian~\cite{zheng2018dags,tang2023fairness} problem:
\begin{equation*}
    || \sum_{i=1}^M \omega_i \nabla_{\theta} \mathcal{F}_i(\theta) ||_2^2 + \sum_{q=1}^Q \lambda_q (\beta_q - \mathbf{h}_q^\top \bm{\omega}) + \rho (1-\sum_{i=1}^M \omega_m),
\end{equation*}
where $\lambda_q \geq 0$ and $\rho\geq 0$. By the KKT conditions associated with this Lagrangian equation, the following inequality holds:
\begin{equation*}
    (\sum_{i=1}^M \omega_i \nabla_{\theta} \mathcal{F}_i(\theta))^\top \nabla_{\theta} \mathcal{F}_m(\theta) = \sum_{q=1}^Q \lambda_q \mathbf{h}_{q,m} + \rho \geq 0,\quad \forall m \in [1,M].
\end{equation*}
This signifies that for any objective function $\mathcal{F}_m$, there exists a gradient 
$\mathbf{g}=\sum_{i=1}^M \omega_i \nabla_{\theta} \mathcal{F}_i(\theta)$ such that $\mathbf{g}^\top \nabla_\theta \mathcal{F}_m(\theta)\geq 0$, indicating that the solution finds a direction that is able to increase all objectives.
\end{proof}

\subsection{Overall Objective}
By combining the above designs, the overall training loss function of our method is formulated as follows:
\begin{equation}
    \label{eq:final_loss_fun}
    \mathcal{L}_\theta = - \sum_{i=1}^M \omega_i \mathbb{E}_{\pi_{\theta_{b}}} \left[ 
        \min\left( \frac{\pi_\theta(\hat{\mathbf{x}}_{u,v}|u,v)}{\pi_{\theta_{b}} (\hat{\mathbf{x}}_{u,v}|u,v)} A_i(\hat{\mathbf{x}}_{u,v}), \text{clip}\left(\frac{\pi_\theta(\hat{\mathbf{x}}_{u,v}|u,v)}{\pi_{\theta_{b}}(\hat{\mathbf{x}}_{u,v}|u,v)} ,1-\epsilon, 1+\epsilon\right)
        A_i({\hat{\mathbf{x}}_{u,v}})\right)\right],
\end{equation}
where $A_i(\hat{\mathbf{x}}_{u,v})$ represents the advantage value of the generated explanation $\hat{\mathbf{x}}_{u,v}$ in terms of the $i$-th quality perspective, $w_i$ denotes the weight of the i-th objective loss, and M represents the number of objective functions (e.g., informativeness and persuasiveness). By minimizing the above training objective function, the explainable recommender can optimize personalized recommendation explanations from multiple perspectives, better aligning with human preferences. The complete training process of our framework is shown in Algorithm~\ref{alg}. To begin with, we pretrain the explainable recommendation model using observed data, which is used to initialize the parameters $\theta_b$ of the behavior policy network. Then, in the first phase of our optimization framework, HF4Rec interacts with LLM to generate rewards for predicted explanations, which are then converted into advantage values and stored in the replay buffer. In the second phase, we utilize all the data from the replay buffer to update the target policy network parameters. This procedure is repeated through these two phases until the specified number of iterations is reached. And at the end of each iteration cycle, the parameters $\theta$ of the target policy network overwrite the behavior policy $\theta_b$, reducing the estimation bias caused by distribution shifts. In summary, our method implements interactive optimization framework in an off-policy manner.

\begin{algorithm}[t]
\caption{Learning Algorithm of HF4Rec}
\label{alg}
\KwIn{Observed interaction $\mathcal{D}$; Learning rate $\gamma$; Iteration number $T$; Epoch number $E$; Exploration number of each user-item pair $J$; The retrieved item number in user history $K$; The prior preference vectors $\{\mathbf{h}_1,\mathbf{h}_2,\dots,\mathbf{h}_Q\}$ and scalars $\{\beta_1,\beta_2,\dots,\beta_Q\}$; Human-written prompt prototypes;}
\KwOut{Model Parameters $\theta$}

Initialize the replay buffer $\mathcal{B}$ as $\emptyset$;

Initialize the behavior policy network $\mathcal{M}_{\theta_b}$;

Pretrain $\mathcal{M}_{\theta_b}$ on observed dataset $\mathcal{D}$;

\For{$i \gets 1$ \textbf{to} $T$}{
    \tcp{\textbf{Phase \uppercase\expandafter{\romannumeral1}: Data Collection}}
    Sample unobserved data $\tilde{\mathcal{D}}$ based on Eq.~(\ref{eq:prob_sampling});
    
    \For{$(u,v)$ \textbf{in} $\mathcal{D}\cup\tilde{\mathcal{D}}$}{
        Format item information using Prompt~\ref{prompt:user_item_info};
        
        Retrieve the top-$K$ similar history items based on the information of target item $v$;

        Format user information using Prompt~\ref{prompt:user_item_info};

        Generate Human-Induced Few-shot context using Prompt~\ref{prompt:context_prompt};
        
        \For{$j \gets 1$ \textbf{to} $J$}{
            \tcp{Probabilistic sampling}
            $\hat{\mathbf{x}}_{u,v}^j, \pi_{\theta_b}(\hat{\mathbf{x}}_{u,v}^j|u,v) \gets \mathcal{M}_{\theta_b}(u,v)$;

            Generate multi-perspective reward values $\mathbf{A}(\hat{\mathbf{x}}_{u,v}^j)$ using Prompt~\ref{prompt:reward_prompt};
            
            Store $(u,v,\hat{\mathbf{x}}_{u,v}^j,\mathbf{A}(\hat{\mathbf{x}}_{u,v}^j),\pi_{\theta_b}(\hat{\mathbf{x}}_{u,v}^j|u,v))$ into $\mathcal{B}$;
        }
    }
    Replace all reward values in $\mathcal{B}$ with advantage values using Eq.~(\ref{eq:adv_value});

    \tcp{\textbf{Phase \uppercase\expandafter{\romannumeral2}: Model Updating}}

    Initialize the target policy network $\mathcal{M}_{\theta}$ with $\theta_b$;

    \For{$e \gets 1$ \textbf{to} $E$}{
        Calculate the weights $\bm{\omega}$ using Eq.~(\ref{eq:weights});

        Update $\theta$ using  Eq.~(\ref{eq:final_loss_fun}) on dataset $\mathcal{B}$;
    }

    Replace behavior policy parameters $\theta_b$ with $\theta$; 
}

\KwRet{$\theta$};
\end{algorithm}

\section{EXPERIMENTS}
In this section, we conduct extensive experiments to demonstrate the effectiveness of out optimization framework, where we focus on the following research questions:
\begin{itemize}
    \item \textbf{RQ1:} How does our proposed framework perform across different explainable recommendation backbones in terms of recommendation accuracy, objective text metrics, and subjective human evaluations?
    \item \textbf{RQ2:} How do the different components within our framework contribute to the final performance?
    \item \textbf{RQ3:} Whether large language models, as human simulators, can effectively substitute for human participation in our model optimization and evaluation?
    \item \textbf{RQ4:} How do key hyper-parameters within our framework (\emph{e.g.}, the training data volume, the exploration number of generated explanations) impact the overall model performance?
    \item \textbf{RQ5:} Whether the Pareto optimization in HF4Rec can enhance the multi-perspective quality of explanations and achieve the desired Pareto front results?
    \item \textbf{RQ6:} Whether HF4Rec can improve the model robustness in the face of data sparsity?
\end{itemize}

\subsection{Experimental Setup}
\subsubsection{Datasets}
To validate the effectiveness of our proposed method, we conduct extensive experiments on user interaction behaviors across different types of product domains, as testing on diverse datasets is more beneficial for verifying the generalizability of our proposed algorithm. Specifically, the dataset details are as follows:
\begin{itemize}
    \item \textit{\textbf{Amazon Beauty, Sports, and VideoGames}}\footnote{\url{https://cseweb.ucsd.edu/~jmcauley/datasets/amazon/links.html}}: These datasets consists of product reviews and metadata from the Amazon website. We select three specific subcategories: \textit{Beauty, Sports and Outdoors, and Video Games}.
    \item \textit{\textbf{Yelp}}\footnote{\url{https://www.yelp.com/dataset}}: This widely-used dataset for business recommendations is collected from the Yelp website and contains numerous user comments. In this dataset, business venues that user visited are treated as items.
\end{itemize}

For the Amazon datasets, we filter out items and users with fewer than five interactions (\emph{i.e.}, 5-core), and utilize the description, category, and title as item features. For the Yelp dataset, due to its larger scale, we adopt a 15-core filtering approach, and use the state abbreviation, city, name, and category as item features. We leave the last user interaction as the test sub-dataset, while the remaining interactions serving as the training sub-dataset. The detailed statistics of the above datasets are shown in Table~\ref{tab:datasets}. These datasets vary in different domains, scale, and density, helping to demonstrate the generality of the proposed framework.

\subsubsection{Baselines} 
Given that the explainable recommendation task includes both explanation generation and rating prediction, we will introduce the baselines used for two tasks, respectively. 

\begin{table}[t]
    \centering
    \caption{Statistics of the datasets.}
    \label{tab:datasets}
    \begin{tabular}{c*{4}{|c}} 
        \hline\hline
         & Beauty & Sports & VideoGames & Yelp\\
         \hline
         \# User & 5,396 & 9,340 & 13,957 & 15,025 \\
         \# Item & 3,178 & 5,904 & 7,378 & 12,445 \\
         \# Inter. & 54,805 & 73,864 & 140,353 & 698,084 \\
         \# Feature & 944 & 2,049 & 2,446 & 5,619 \\
         Sparsity & 99.68\% & 99.87\% & 99.86\% & 99.63\% \\
         Avg. Length & 11.71 & 12.68 & 13.13 & 12.05 \\
        \hline\hline
    \end{tabular}
\end{table}

\begin{enumerate}[leftmargin=*]
    \item \textbf{Explanation Generation Task.} To evaluate the quality of generated explanations, we compare our framework with the following representative baselines:
    \begin{itemize}
        \item \textbf{PETER}~\cite{li-etal-2021-personalized} utilizes a Transformer language model to predict target sentences based on given user and item IDs, infusing IDs with linguistic meaning to generate personalized explanations.
        \item \textbf{PEPLER}~\cite{li2023personalized} incorporates prompt learning strategy to bridge the gap between continuous prompts and pretrained models, enhancing the performance of explanation generation.
        \item \textbf{ERRA}~\cite{DBLP:conf/acl/ChengWLZ0LL23} enhances recommendation explanations by integrating retrieval-augmented information from the training set and aspect-enhancement components into the model, producing more accurate explanations.
    \end{itemize}

Based on the three backbones, we propose two methods:  \textbf{(1) Holistic-based Method}: We modify the multi-perspective reward based on informativeness and persuasiveness in Prompt~\ref{prompt:reward_prompt} to generate a reward focused on overall quality of generated recommendation explanation. The corresponding backbones are denoted as \textbf{PETER-H}, \textbf{PEPLER-H}, and \textbf{ERRA-H}. \textbf{(2) Multi-perspective-based Method}: We use the multi-perspective reward based on informativeness and persuasiveness (\emph{cf.} Prompt~\ref{prompt:reward_prompt}), with the objective function Eq.~\eqref{eq:final_loss_fun} in Sec.~\ref{sec:overall_obj} as the model optimization goal. The corresponding backbones are denoted as \textbf{PETER-M}, \textbf{PEPLER-M}, and \textbf{ERRA-M}.

    \item \textbf{Rating Prediction Task.} 
    For the rating prediction task, we compare our framework against the following two simple yet effective baseline models:
    \begin{itemize}
        \item \textbf{PMF}~\cite{mnih2007probabilistic} is a well-known matrix factorization method that represents users and items as latent vectors, and the rating predictions are computed using the dot product of these vectors.
        \item \textbf{SVD++}~\cite{koren2010factor} enhances the traditional Singular Value Decomposition (SVD)~\cite{zhou2015svd} approach by incorporating implicit feedback to better capture and predict user preferences.
        \item \textbf{NRT}~\cite{li2017neural} proposes a deep learning-based framework that improves rating prediction quality by simultaneously optimizing user tips and rating predictions.
        \item \textbf{NETE}~\cite{li2020generate} proposes a learning framework for neural template explanation generation and rating prediction, enhancing recommendation performance by incorporating rich contextual information.
    \end{itemize}
\end{enumerate}

\subsubsection{Evaluation Metrics}
To comprehensively evaluate the performance of explanation generation, we assess the  experimental results from both objective and subjective aspects:
\begin{itemize}[leftmargin=*]
    \item \textbf{Objective Evaluation}. Following the previous work~\cite{li-etal-2021-personalized,li2023personalized,DBLP:conf/acl/ChengWLZ0LL23,geng2022recommendation,li2023prompt,li2023ucepic}, we employ commonly adopted text similarity metrics, including BERT-Score~\cite{bert-score}, BLEU~\cite{Papineni_Roukos_Ward_Zhu_2001}, and ROUGE~\cite{lin2004rouge}. Specifically, we report precision, recall, and F1 scores for ROUGE-1, ROUGE-2, and ROUGE-L, which correspond to the smoothed ROUGE scores for 1-grams, 2-grams, and longest common subsequence, respectively. Moreover, in the field of explainable recommendation, the textual similarity does not necessarily equal to explainability. For example, humans may focus more on whether the explanation includes specific item features. Hence, we introduce the Feature Matching Rate (FMR)~\cite{li-etal-2021-personalized} metric, which represents the coverage rate of features between prediction and oracle texts.
    \item \textbf{Subjective Evaluation}. To further explore whether generated explanations can indeed assist users, we design questionnaires to ask the feelings of the users on explanation effectiveness. Specifically, we focus on two perspectives of the recommendation explanations: (1) \textit{Informativeness: whether the generated explanation help users learn more about the product's characteristic or usage experience?} (2) \textit{Persuasiveness: whether the generated explanation can aid users in making fast decisions?}
    For each aspect, participants are asked to assess on a scale of 1 to 3, where higher scores indicate better performance. To save human resource costs, in this paper, we employ the ChatGPT-3.5 Turbo\footnote{\url{https://platform.openai.com/docs/models/overview}} as human proxies to evaluate 200 samples randomly collected from each dataset. And the empirical analysis and discussion with respect to the rationality of substituting LLM for real-human in subjective evaluation experiment are explored in Section~\ref{sec:consistency}.
    \item \textbf{Recommendation Performance Evaluation}. For measuring the accuracy of the rating prediction task, we use Root Mean Square Error (RMSE) and Mean Absolute Error (MAE) as the evaluation metrics, where lower values signify better recommendation performance.
\end{itemize}

\subsubsection{Implementation Details}
For all the compared models, we use the open-source code provided by the respective authors and reuse the default settings specified in their original publications. For our methods, we implement them based on the PyTorch~\cite{paszke2019pytorch} platform. To mitigate the problem of Out-Of-Vocabulary (OOV) words, we employ the Byte Pair Encoding (BPE) tokenizer from GPT-2 to construct the vocabulary, and we limit the maximum length of the generated explanations at 15 BPE tokens. 

To ensure fairness, we fix the batch size at 128 and align the embedding size with the optimal parameters of the original backbone models. We optimize all models using the Adam~\cite{kingma2014adam} optimizer, and initialize model parameters with the uniform distribution. We search the learning rate $\gamma$ from $\{1e\text{-}5,5e\text{-}5,1e\text{-}4,5e\text{-}4,1e\text{-}3\}$, and tune clipping hyper-parameter $\epsilon$ in the range of $\{0.1,0.2,0.3\}$. We apply the $K$-means method to cluster users into five groups, and for each representative user in the groups, we randomly select an interacted item and manually write three scoring examples as prompt prototypes. We set each prior preference vector and scalar value in Pareto optimization to one-hot vector and 0.2, respectively. For the number of sampled unobserved data and the number of customized examples in data collection phase, we set them at 6400 and 3 respectively.

\subsection{Overall Performance (RQ1)}

\begin{table}[htbp]
    \centering
    \caption{Performance comparison of explanation quality in terms of Objective Evaluation. BS-P, BS-R, BS-F, R1-P, R1-R, R1-F, R2-P, R2-R, R2-F, RL-P, RL-R, RL-F denote Precision, Recall and F1 of BERTScore, ROUGE-1, ROUGE-2 and ROUGE-L, respectively. To facilitate performance comparison, all metrics in the table are shown in percentage values (\% symbol is omitted for clearer presentation). The best results are highlighted in \textbf{bold}. For better visual effect, we distinguished the best-performing results of the \colorbox{SkyBlue}{holistic-based methods} and \colorbox{Peach}{multi-perspective-based methods} using different colors. Our framework’s improvements exhibit statistical significant under a paired t-test with a significance level of $p$ < 0.05.}
    \label{tab:obj}
    \renewcommand{\arraystretch}{1.2}
    \resizebox{\textwidth}{!}{
    \begin{tabular}{*{15}{c}} 
        \hline\hline
        \multicolumn{15}{c}{Beauty}\\
         & FMR & BS-P & BS-R & BS-F & BLEU & R1-P & R1-R & R1-F & R2-P & R2-R & R2-F & RL-P & RL-R & RL-F\\ 
        \toprule 
        PETER    & 20.682          & 84.498          & 84.889          & 84.676          & 11.770          & 13.679          & { 14.760}    & 13.306          & 2.354          & { 2.588}    & { 2.273}    & 12.153          & { 13.232}    & 11.843          \\
PETER-H  & \colorbox{SkyBlue}{\textbf{23.239}} & { 84.962}    & \colorbox{SkyBlue}{\textbf{85.163}} & { 85.045}    & \colorbox{SkyBlue}{\textbf{12.672}} & \colorbox{SkyBlue}{\textbf{14.815}} & \colorbox{SkyBlue}{\textbf{16.042}} & \colorbox{SkyBlue}{\textbf{14.481}} & \colorbox{SkyBlue}{\textbf{2.647}} & \colorbox{SkyBlue}{\textbf{2.961}} & \colorbox{SkyBlue}{\textbf{2.583}} & \colorbox{SkyBlue}{\textbf{13.098}} & \colorbox{SkyBlue}{\textbf{14.345}} & \colorbox{SkyBlue}{\textbf{12.845}} \\
PETER-M  & { 22.646}    & \colorbox{Peach}{\textbf{85.251}} & { 85.113}    & \colorbox{Peach}{\textbf{85.165}} & { 12.059}    & { 14.438}    & 14.480          & { 13.473}    & { 2.406}    & 2.357          & 2.176          & { 12.979}    & 13.055          & { 12.098}    \\
\hline
PEPLER   & 20.960          & 84.122          & 85.152          & 84.620          & 10.659          & 11.912          & { 15.987}    & 12.987          & { 1.504}    & 2.170          & { 1.665}    & 10.433          & { 14.227}    & 11.426          \\
PEPLER-H & { 21.349}    & { 84.788}    & { 85.435}    & { 85.097}    & \colorbox{SkyBlue}{\textbf{11.399}} & { 12.791}    & \colorbox{SkyBlue}{\textbf{18.674}} & \colorbox{SkyBlue}{\textbf{14.467}} & 1.380          & { 2.204}    & 1.600          & { 10.734}    & \colorbox{SkyBlue}{\textbf{15.992}} & \colorbox{SkyBlue}{\textbf{12.210}} \\
PEPLER-M & \colorbox{Peach}{\textbf{24.889}} & \colorbox{Peach}{\textbf{85.324}} & \colorbox{Peach}{\textbf{85.630}} & \colorbox{Peach}{\textbf{85.460}} & { 11.373}    & \colorbox{Peach}{\textbf{13.850}} & 14.935          & { 13.481}    & \colorbox{Peach}{\textbf{2.016}} & \colorbox{Peach}{\textbf{2.263}} & \colorbox{Peach}{\textbf{1.992}} & \colorbox{Peach}{\textbf{12.459}} & 13.515          & { 12.147}    \\
\hline
ERRA     & 23.443          & { 85.909}    & 85.268          & 85.568          & 12.726          & 15.960          & 15.068          & 14.361          & 2.825          & 2.720          & 2.528          & 14.364          & 13.539          & 12.886          \\
ERRA-H   & { 25.556}    & 86.535          & { 85.523}    & { 86.007}    & { 13.299}    & { 17.276}    & { 16.055}    & { 15.431}    & \colorbox{SkyBlue}{\textbf{3.127}} & { 2.924}    & { 2.755}    & \colorbox{SkyBlue}{\textbf{15.545}} & { 14.465}    & { 13.863}    \\
ERRA-M   & \colorbox{Peach}{\textbf{30.393}} & \colorbox{Peach}{\textbf{86.581}} & \colorbox{Peach}{\textbf{85.897}} & \colorbox{Peach}{\textbf{86.222}} & \colorbox{Peach}{\textbf{15.170}} & \colorbox{Peach}{\textbf{17.444}} & \colorbox{Peach}{\textbf{19.181}} & \colorbox{Peach}{\textbf{17.131}} & { 2.978}    & \colorbox{Peach}{\textbf{3.361}} & \colorbox{Peach}{\textbf{2.935}} & { 15.281}    & \colorbox{Peach}{\textbf{16.996}} & \colorbox{Peach}{\textbf{15.053}}\\
        \hline\hline
        \multicolumn{15}{c}{Sports}\\
         & FMR & BS-P & BS-R & BS-F & BLEU & R1-P & R1-R & R1-F & R2-P & R2-R & R2-F & RL-P & RL-R & RL-F\\ 
         \hline
         PETER    & 8.619           & 83.878          & 84.109          & 83.975          & 9.725           & 11.377          & 11.304          & 10.603          & 1.283          & 1.341          & 1.206          & 10.205          & 10.174          & 9.500           \\
PETER-H  & { 9.829}     & \colorbox{SkyBlue}{\textbf{84.444}} & { 84.208}    & { 84.306}    & { 9.853}     & { 12.584}    & { 11.860}    & { 11.291}    & \colorbox{SkyBlue}{\textbf{1.576}} & { 1.577}    & \colorbox{SkyBlue}{\textbf{1.433}} & \colorbox{SkyBlue}{\textbf{11.322}} & { 10.696}    & { 10.138}    \\
PETER-M  & \colorbox{Peach}{\textbf{10.150}} & { 84.260}    & \colorbox{Peach}{\textbf{84.405}} & \colorbox{Peach}{\textbf{84.317}} & \colorbox{Peach}{\textbf{11.178}} & \colorbox{Peach}{\textbf{12.670}} & \colorbox{Peach}{\textbf{13.744}} & \colorbox{Peach}{\textbf{12.392}} & { 1.359}    & \colorbox{Peach}{\textbf{1.670}} & { 1.393}    & { 11.181}    & \colorbox{Peach}{\textbf{12.272}} & \colorbox{Peach}{\textbf{10.973}} \\
\hline
PEPLER   & 7.998           & 83.636          & 84.518          & 84.061          & 9.347           & 10.357          & 12.932          & 10.991          & 0.916          & 1.283          & 1.005          & 8.987           & 11.375          & 9.573           \\
PEPLER-H & { 9.176}     & { 84.399}    & { 84.857}    & { 84.615}    & { 10.661}    & { 11.644}    & \colorbox{SkyBlue}{\textbf{16.781}} & { 13.146}    & { 1.011}    & \colorbox{SkyBlue}{\textbf{1.715}} & { 1.197}    & { 9.880}     & \colorbox{SkyBlue}{\textbf{14.530}} & { 11.222}    \\
PEPLER-M & \colorbox{Peach}{\textbf{11.017}} & \colorbox{Peach}{\textbf{86.050}} & \colorbox{Peach}{\textbf{85.249}} & \colorbox{Peach}{\textbf{85.634}} & \colorbox{Peach}{\textbf{11.783}} & \colorbox{Peach}{\textbf{14.226}} & { 14.397}    & \colorbox{Peach}{\textbf{13.449}} & \colorbox{Peach}{\textbf{1.373}} & { 1.440}    & \colorbox{Peach}{\textbf{1.309}} & \colorbox{Peach}{\textbf{12.408}} & { 12.649}    & \colorbox{Peach}{\textbf{11.750}} \\
\hline
ERRA     & 9.711           & 85.229          & 84.533          & 84.862          & 10.334          & 13.454          & 12.133          & 11.875          & { 1.749}    & 1.583          & { 1.518}    & 12.176          & 11.004          & 10.734          \\
ERRA-H   & { 10.557}    & \colorbox{SkyBlue}{\textbf{85.566}} & \colorbox{SkyBlue}{\textbf{84.648}} & \colorbox{SkyBlue}{\textbf{85.088}} & { 10.691}    & { 14.211}    & { 12.807}    & { 12.523}    & \colorbox{SkyBlue}{\textbf{1.857}} & \colorbox{SkyBlue}{\textbf{1.804}} & \colorbox{SkyBlue}{\textbf{1.649}} & \colorbox{SkyBlue}{\textbf{12.686}} & { 11.495}    & { 11.180}    \\
ERRA-M   & \colorbox{Peach}{\textbf{11.467}} & { 85.439}    & { 84.557}    & { 84.979}    & \colorbox{Peach}{\textbf{10.737}} & \colorbox{Peach}{\textbf{14.222}} & \colorbox{Peach}{\textbf{12.917}} & \colorbox{Peach}{\textbf{12.567}} & 1.594          & { 1.688}    & 1.483          & { 12.677}    & \colorbox{Peach}{\textbf{11.568}} & \colorbox{Peach}{\textbf{11.201}}\\
        \hline\hline 
        \multicolumn{15}{c}{VideoGames}\\
         & FMR & BS-P & BS-R & BS-F & BLEU & R1-P & R1-R & R1-F & R2-P & R2-R & R2-F & RL-P & RL-R & RL-F\\ 
        \toprule 
        PETER    & 26.918          & 84.349          & 84.827          & 84.570          & 11.370          & 13.062          & 14.609          & 13.010          & 2.012          & 2.530          & 2.072          & 11.523          & 13.031          & 11.513          \\
PETER-H  & { 29.777}    & { 85.585}    & { 85.210}    & { 85.378}    & { 12.618}    & { 15.634}    & { 15.564}    & { 14.469}    & { 2.682}    & { 2.926}    & { 2.532}    & { 13.915}    & { 13.920}    & { 12.882}    \\
PETER-M  & \colorbox{Peach}{\textbf{33.001}} & \colorbox{Peach}{\textbf{86.439}} & \colorbox{Peach}{\textbf{85.434}} & \colorbox{Peach}{\textbf{85.914}} & \colorbox{Peach}{\textbf{12.659}} & \colorbox{Peach}{\textbf{17.371}} & \colorbox{Peach}{\textbf{15.934}} & \colorbox{Peach}{\textbf{15.284}} & \colorbox{Peach}{\textbf{3.112}} & \colorbox{Peach}{\textbf{3.119}} & \colorbox{Peach}{\textbf{2.793}} & \colorbox{Peach}{\textbf{15.509}} & \colorbox{Peach}{\textbf{14.261}} & \colorbox{Peach}{\textbf{13.637}} \\
\hline
PEPLER   & 26.990          & 84.354          & 85.259          & 84.789          & 11.271          & 12.801          & 16.583          & 13.765          & 1.748          & 2.586          & 1.948          & 10.964          & 14.401          & 11.841          \\
PEPLER-H & \colorbox{SkyBlue}{\textbf{30.257}} & \colorbox{SkyBlue}{\textbf{84.849}} & \colorbox{SkyBlue}{\textbf{85.576}} & \colorbox{SkyBlue}{\textbf{85.196}} & \colorbox{SkyBlue}{\textbf{12.405}} & \colorbox{SkyBlue}{\textbf{13.851}} & \colorbox{SkyBlue}{\textbf{18.745}} & \colorbox{SkyBlue}{\textbf{15.186}} & \colorbox{SkyBlue}{\textbf{2.018}} & \colorbox{SkyBlue}{\textbf{3.190}} & \colorbox{SkyBlue}{\textbf{2.313}} & \colorbox{SkyBlue}{\textbf{11.838}} & \colorbox{SkyBlue}{\textbf{16.249}} & \colorbox{SkyBlue}{\textbf{13.036}} \\
PEPLER-M & { 28.817}    & { 84.726}    & { 85.558}    & { 85.125}    & { 11.801}    & { 13.395}    & { 18.217}    & { 14.725}    & { 1.932}    & { 3.038}    & { 2.207}    & { 11.444}    & { 15.803}    & { 12.640}    \\
\hline
ERRA     & 33.188          & 84.872          & 85.395          & 85.116          & 12.701          & 14.446          & 17.562          & 15.041          & 2.411          & 3.242          & 2.578          & 12.556          & 15.446          & 13.120          \\
ERRA-H   & { 35.767}    & \colorbox{SkyBlue}{\textbf{85.743}} & { 85.620}    & { 85.665}    & { 13.939}    & { 15.793}    & { 18.723}    & { 16.146}    & { 2.717}    & \colorbox{SkyBlue}{\textbf{3.646}} & { 2.869}    & { 13.757}    & { 16.495}    & { 14.107}    \\
ERRA-M   & \colorbox{Peach}{\textbf{35.867}} & { 85.706}    & \colorbox{Peach}{\textbf{85.660}} & \colorbox{Peach}{\textbf{85.667}} & \colorbox{Peach}{\textbf{13.982}} & \colorbox{Peach}{\textbf{15.905}} & \colorbox{Peach}{\textbf{18.839}} & \colorbox{Peach}{\textbf{16.291}} & \colorbox{Peach}{\textbf{2.770}} & { 3.637}    & \colorbox{Peach}{\textbf{2.908}} & \colorbox{Peach}{\textbf{13.800}} & \colorbox{Peach}{\textbf{16.519}} & \colorbox{Peach}{\textbf{14.179}}\\
\hline\hline
        \multicolumn{15}{c}{Yelp}\\
         & FMR & BS-P & BS-R & BS-F & BLEU & R1-P & R1-R & R1-F & R2-P & R2-R & R2-F & RL-P & RL-R & RL-F\\ 
        \toprule 
        PETER    & 9.151           & 83.459          & 84.252          & 83.836          & 11.849          & 13.039          & 15.619          & 13.614          & 1.344          & 1.720          & 1.435          & 11.682          & 14.120          & 12.236          \\
PETER-H  & { 10.323}    & { 83.636}    & { 84.300}    & { 83.949}    & { 12.015}    & { 13.249}    & { 16.089}    & { 13.895}    & { 1.429}    & { 1.922}    & { 1.549}    & { 11.936}    & { 14.594}    & { 12.552}    \\
PETER-M  & \colorbox{Peach}{\textbf{12.632}} & \colorbox{Peach}{\textbf{84.340}} & \colorbox{Peach}{\textbf{84.562}} & \colorbox{Peach}{\textbf{84.435}} & \colorbox{Peach}{\textbf{13.009}} & \colorbox{Peach}{\textbf{13.872}} & \colorbox{Peach}{\textbf{17.824}} & \colorbox{Peach}{\textbf{14.889}} & \colorbox{Peach}{\textbf{1.708}} & \colorbox{Peach}{\textbf{2.439}} & \colorbox{Peach}{\textbf{1.891}} & \colorbox{Peach}{\textbf{12.666}} & \colorbox{Peach}{\textbf{16.393}} & \colorbox{Peach}{\textbf{13.636}} \\
\hline
PEPLER   & 7.953           & 83.637          & 84.543          & 84.072          & 10.715          & 11.835          & 15.161          & 12.745          & 1.140          & 1.600          & 1.264          & 10.501          & 13.565          & 11.340          \\
PEPLER-H & { 10.110}    & { 84.159}    & \colorbox{SkyBlue}{\textbf{84.806}} & { 84.465}    & \colorbox{SkyBlue}{\textbf{12.061}} & \colorbox{SkyBlue}{\textbf{13.107}} & \colorbox{SkyBlue}{\textbf{17.391}} & \colorbox{SkyBlue}{\textbf{14.322}} & \colorbox{SkyBlue}{\textbf{1.440}} & { 2.087}    & { 1.620}    & \colorbox{SkyBlue}{\textbf{11.515}} & { 15.432}    & \colorbox{SkyBlue}{\textbf{12.625}} \\
PEPLER-M & \colorbox{Peach}{\textbf{11.634}} & \colorbox{Peach}{\textbf{84.397}} & { 84.643}    & \colorbox{Peach}{\textbf{84.504}} & { 11.872}    & { 12.650}    & { 17.622}    & { 14.099}    & { 1.438}    & \colorbox{Peach}{\textbf{2.284}} & \colorbox{Peach}{\textbf{1.667}} & { 11.169}    & \colorbox{Peach}{\textbf{15.683}} & { 12.481}    \\
\hline
ERRA     & 8.825           & 84.338          & 84.443          & 84.375          & 12.249          & 13.491          & 16.871          & 14.363          & 1.445          & 1.941          & 1.573          & 11.979          & 15.159          & 12.812          \\
ERRA-H   & { 11.740}    & { 84.860}    & \colorbox{SkyBlue}{\textbf{84.606}} & { 84.718}    & { 12.725}    & { 13.552}    & \colorbox{SkyBlue}{\textbf{18.597}} & { 15.005}    & { 1.705}    & { 2.596}    & { 1.950}    & { 12.437}    & \colorbox{SkyBlue}{\textbf{17.209}} & { 13.816}    \\
ERRA-M   & \colorbox{Peach}{\textbf{12.093}} & \colorbox{Peach}{\textbf{84.987}} & { 84.522}    & \colorbox{Peach}{\textbf{84.739}} & \colorbox{Peach}{\textbf{13.094}} & \colorbox{Peach}{\textbf{13.877}} & { 18.472}    & \colorbox{Peach}{\textbf{15.141}} & \colorbox{Peach}{\textbf{1.873}} & \colorbox{Peach}{\textbf{2.754}} & \colorbox{Peach}{\textbf{2.103}} & \colorbox{Peach}{\textbf{12.790}} & { 17.133}    & \colorbox{Peach}{\textbf{13.992}}\\
        \hline\hline
    \end{tabular}
}
\end{table}

\begin{table}[htbp]
    \centering
    \caption{Performance comparison of explanation quality in terms of Subjective Evaluation. ``Info.'' and ``Persv.'' stand for Informativeness and Persuasiveness respectively. The best results are highlighted in \textbf{bold}.  For better visual effect, we distinguished the best-performing results of the \colorbox{SkyBlue}{holistic-based methods} and \colorbox{Peach}{multi-perspective-based methods} using different colors.}
    \label{tab:sub}
    \renewcommand{\arraystretch}{1.15}
    \resizebox{\textwidth}{!}{
    \begin{tabular}{*{3}{c}|*{2}{c}|*{2}{c}|*{2}{c}} 
    \hline\hline
     & \multicolumn{2}{c|}{Beauty} & \multicolumn{2}{c|}{Sports} & \multicolumn{2}{c|}{VideoGames} & \multicolumn{2}{c}{Yelp}\\
     & Info. & Persv. & Info. & Persv. &Info. & Persv. & Info. & Persv. \\
     \toprule
     PETER    & 1.460±0.5839          & 1.775±0.6687          & 1.510±0.4999          & 1.655±0.5709          & 1.275±0.4465          & 1.485±0.5097          & 1.770±0.4440          & 1.700±0.6164          \\
PETER-H  & { 1.540±0.5731}    & { 1.840±0.6359}    & { 1.610±0.4979}    & { 1.830±0.5302}    & { 1.480±0.4996}    & { 1.780±0.4377}    & { 1.895±0.4049}    & { 1.865±0.6138}    \\
PETER-M  & \colorbox{Peach}{\textbf{1.640±0.5481}} & \colorbox{Peach}{\textbf{2.010±0.5999}} & \colorbox{Peach}{\textbf{1.675±0.4789}} & \colorbox{Peach}{\textbf{1.905±0.6050}} & \colorbox{Peach}{\textbf{1.565±0.5057}} & \colorbox{Peach}{\textbf{1.875±0.3733}} & \colorbox{Peach}{\textbf{1.925±0.2990}} & \colorbox{Peach}{\textbf{2.000±0.4123}} \\
\hline
PEPLER   & 1.630±0.5414          & 1.830±0.6009          & 1.735±0.4948          & 1.790±0.6131          & 1.445±0.4970          & 1.695±0.5118          & 1.795±0.5856          & 1.800±0.6403          \\
PEPLER-H & { 1.885±0.5760}    & { 2.250±0.5895}    & \colorbox{SkyBlue}{\textbf{2.235±0.4688}} & \colorbox{SkyBlue}{\textbf{2.320±0.4976}} & \colorbox{SkyBlue}{\textbf{1.820±0.4331}} & \colorbox{SkyBlue}{\textbf{1.920±0.3790}} & { 2.020±0.3600}    & { 2.040±0.5817}    \\
PEPLER-M & \colorbox{Peach}{\textbf{1.955±0.6580}} & \colorbox{Peach}{\textbf{2.390±0.5458}} & { 2.000±0.3162}    & { 2.160±0.4055}    & { 1.720±0.4490}    & { 1.885±0.4022}    & \colorbox{Peach}{\textbf{2.160±0.4055}} & \colorbox{Peach}{\textbf{2.250±0.5810}} \\
\hline
ERRA     & 1.595±0.5486          & 1.915±0.5981          & 1.610±0.5366          & 1.780±0.5582          & 1.345±0.4754          & 1.685±0.4752          & 1.925±0.4576          & 1.875±0.6159          \\
ERRA-H   & { 1.845±0.4012}    & { 2.055±0.5020}    & { 1.785±0.4677}    & { 2.025±0.4737}    & \colorbox{SkyBlue}{\textbf{1.630±0.4828}} & { 1.875±0.3597}    & { 2.055±0.3633}    & { 2.160±0.5517}    \\
ERRA-M   & \colorbox{Peach}{\textbf{1.870±0.3648}} & \colorbox{Peach}{\textbf{2.195±0.5263}} & \colorbox{Peach}{\textbf{1.825±0.4054}} & \colorbox{Peach}{\textbf{2.115±0.4917}} & {1.595±0.49090}   & \colorbox{Peach}{\textbf{1.880±0.3945}} & \colorbox{Peach}{\textbf{2.065±0.3616}} & \colorbox{Peach}{\textbf{2.210±0.4538}}\\
    \hline\hline
    \end{tabular}
}
\end{table}

\begin{table}[htbp]
    \centering
    \caption{Performance comparison of different methods on four datasets in terms of Recommendation Accuracy. The best results are highlighted in \textbf{bold}.  For better visual effect, we distinguished the best-performing results of the \colorbox{SkyBlue}{holistic-based methods} and \colorbox{Peach}{multi-perspective-based methods} using different colors.}
    \label{tab:acc}
    \renewcommand{\arraystretch}{1.15}
    \begin{tabular}{*{3}{c}|*{2}{c}|*{2}{c}|*{2}{c}} 
    \hline\hline
     & \multicolumn{2}{c|}{Beauty} & \multicolumn{2}{c|}{Sports} & \multicolumn{2}{c|}{VideoGames} & \multicolumn{2}{c}{Yelp}\\
     & RMSE & MAE & RMSE & MAE & RMSE & MAE & RMSE & MAE \\
     \toprule
    PMF      & 1.3545          & 1.1043          & 1.2497          & 1.0200          & 1.4803          & 1.2083          & 1.6563          & 1.4093          \\
SVD++    & 1.1272 & 0.9542 & 0.9827 & 0.8207 & 1.1386 & 0.9576 & 1.1149 & 0.8666 \\
    NRT & 1.1373 & 0.8125 & 0.9893 & 0.6333 & 1.2500 & 0.8909 & 1.2216 & 0.9154 \\
    NETE & 1.2671 & 0.9502 & 1.0960 & 0.7652 & 1.4444 & 1.0987 & 1.3456 & 1.1011 \\
\hline
PETER    & 1.0928          & { 0.7606}    & 0.9791          & 0.6715          & 1.1221          & 0.8041          & \textbf{1.1149} & { 0.8666}    \\
PETER-H  & { 1.0857}    & 0.7684          & \colorbox{SkyBlue}{\textbf{0.9629}} & { 0.6487}    & \colorbox{SkyBlue}{\textbf{1.1059}} & { 0.7991}    & { 1.1178}    & \colorbox{SkyBlue}{\textbf{0.8643}} \\
PETER-M  & \colorbox{Peach}{\textbf{1.0834}} & \colorbox{Peach}{\textbf{0.7544}} & { 0.9641}    & \colorbox{Peach}{\textbf{0.6468}} & { 1.1102}    & \colorbox{Peach}{\textbf{0.7905}} & 1.1199          & 0.8667          \\
\hline
PEPLER   & { 1.0335}    & \textbf{0.7501} & { 0.9135}    & \textbf{0.6728} & 1.0780          & \textbf{0.8096} & { 1.1052}    & \textbf{0.8606} \\
PEPLER-H & 1.0338          & { 0.7507}    & { 0.9135}    & \colorbox{SkyBlue}{\textbf{0.6728}} & \colorbox{SkyBlue}{\textbf{1.0567}} & { 0.8202}    & \colorbox{SkyBlue}{\textbf{1.1045}} & { 0.8658}    \\
PEPLER-M & \colorbox{Peach}{\textbf{1.0211}} & 0.7589          & \colorbox{Peach}{\textbf{0.9055}} & { 0.6811}    & { 1.0641}    & 0.8374          & { 1.1052}    & \colorbox{Peach}{\textbf{0.8606}} \\
\hline
ERRA     & 1.1210          & 0.7830          & 0.9728          & 0.6516          & 1.1324          & 0.7972          & 1.1686          & 0.8823          \\
ERRA-H   & { 1.0948}    & \colorbox{SkyBlue}{\textbf{0.7672}} & { 0.9471}    & { 0.6298}    & { 1.1197}    & \colorbox{SkyBlue}{\textbf{0.7843}} & \colorbox{SkyBlue}{\textbf{1.1615}} & { 0.8738}    \\
ERRA-M   & \colorbox{Peach}{\textbf{1.0920}} & { 0.7678}    & \colorbox{Peach}{\textbf{0.9370}} & \colorbox{Peach}{\textbf{0.6112}} & \colorbox{Peach}{\textbf{1.1108}} & { 0.7859}    & { 1.1618}    & \colorbox{Peach}{\textbf{0.8683}}\\
    \hline\hline
    \end{tabular}
\end{table}

\begin{figure}[t]
    \centering
    \includegraphics[width=\textwidth]{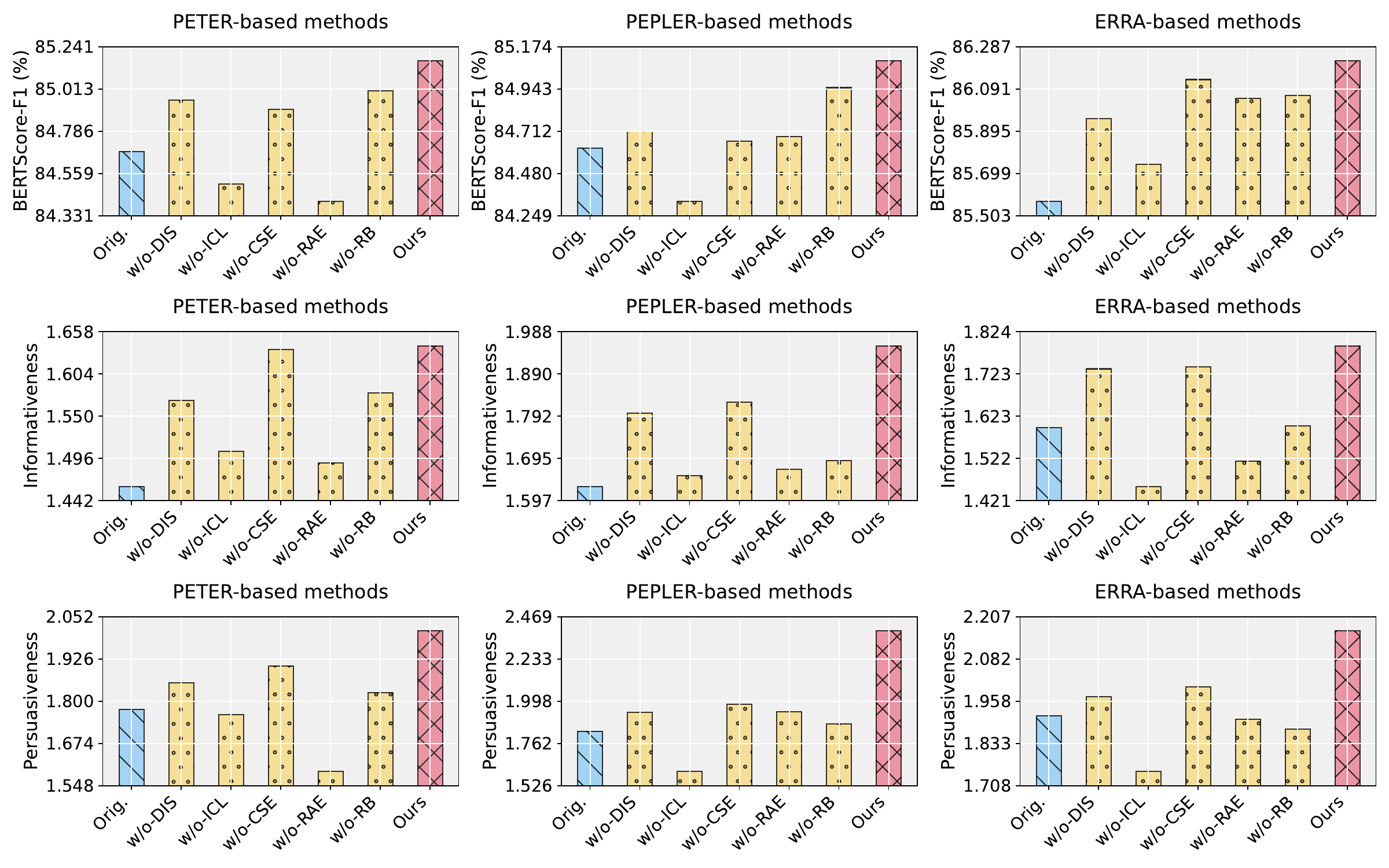}
    \caption{Ablation study on Beauty Dataset.}
    \label{fig:ablation_study_on_beauty}
\end{figure}

\begin{figure}[t]
    \centering
    \includegraphics[width=\textwidth]{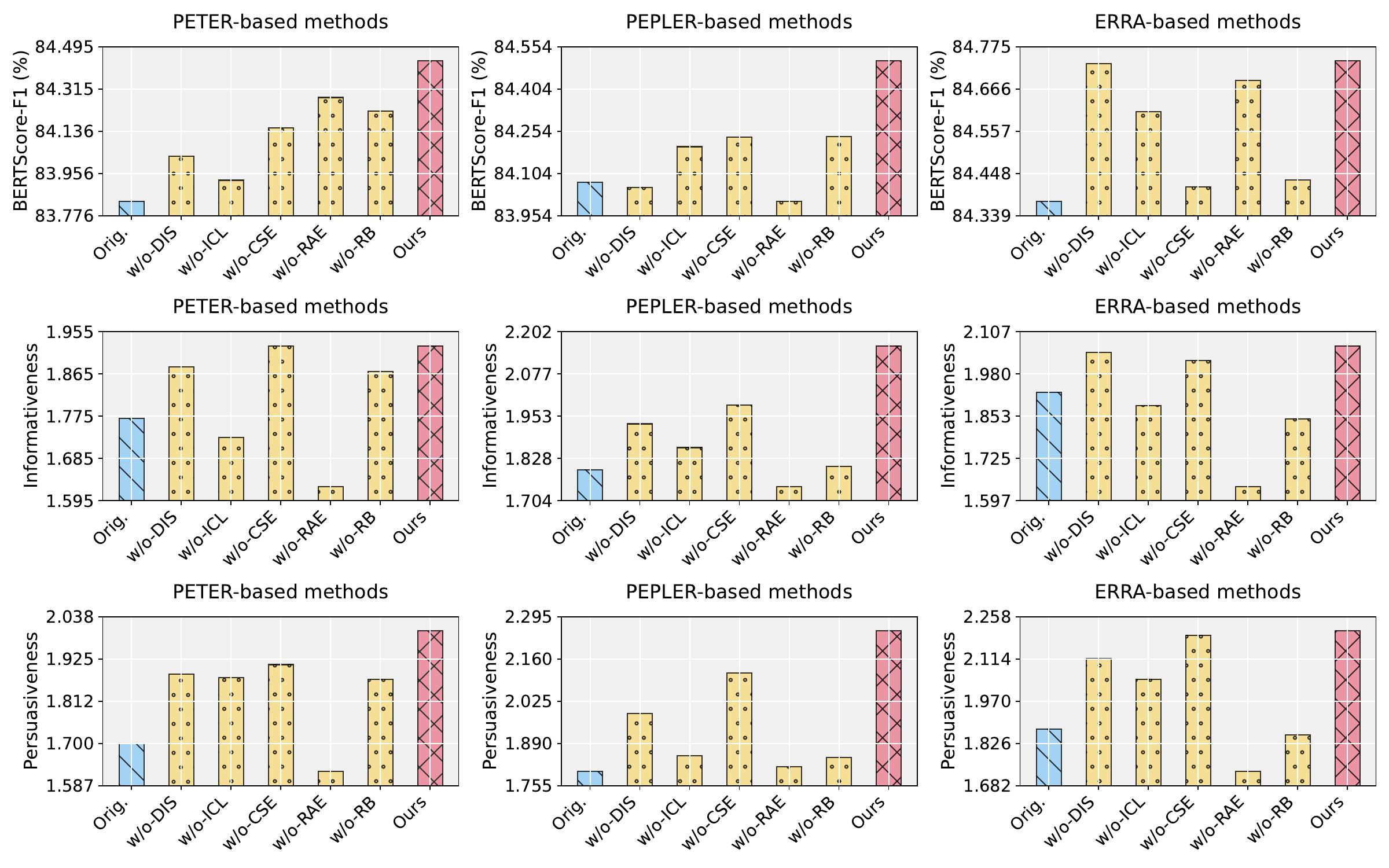}
    \caption{Ablation study on Yelp Dataset.}
    \label{fig:ablation_study_on_yelp}
\end{figure}

We compare the overall performance of all methods across different datasets. Specifically, the objective and subjective evaluation results of explanation quality, and recommendation accuracy performance of rating prediction are presented in Table~\ref{tab:obj},~\ref{tab:sub}, and~\ref{tab:acc}, respectively. Based on the results and our analysis, we make the following key observations             :
\begin{itemize}[leftmargin=*]
    \item Firstly, from Table~\ref{tab:obj}, it is obvious that among all the baselines, ERRA shows relatively better overall performance \emph{w.r.t.} objective evaluation of explanation quality. This underscores the benefit of supplementing additional information with the knowledge base obtained from the training corpus, leading to more accurate and comprehensive text generation. Encouragingly, by applying holistic based (`X-H') and multi-perspective based (`X-M') optimization framework approaches on various explainable recommendation backbones\footnote{`X' represents the base model, \emph{e.g.}, PETER, PEPLER, ERRA, etc.}, the model performances are significantly improved by a large margin across all metrics and datasets. For example, on average, `X-H' methods achieve substantial improvements of 24.31\%, 5.95\%, and 7.25\% in FMR, BS-F, and RL-F, respectively, while `X-M' methods increase of 28.74\%, 8.37\%, and 9.28\%, respectively, which clearly demonstrate the effectiveness and generality of our framework under different settings. We attribute these performance gains to our innovative incorporation of human feedback driven mechanism into the optimization process, enabling the model to delve into additional unobserved augmented data, particularly focused on the long-tailed users and items.
    \item Secondly, as seen from the subjective evaluation results in Table~\ref{tab:sub}, our methods consistently outperform the original base models. These results validate the strong generalization ability of our framework, suggesting that guiding model optimization through anthropomorphic feedback is more suitable for aligning with human preferences than static target fitting.
    Furthermore, multi-perspective-based models play a more significant role in enhancing explanation quality compared with holistic-based models in most cases (19 out of 24). One potential reason is that optimization approach focused on holistic quality rewards may struggle to enhance different aspects of quality simultaneously due to negative transfer influence between conflicting objectives, leading to suboptimal explanation performance. In contrast, the multi-perspective framework based on Pareto optimality effectively decouples different optimization objectives, thereby achieving superior informativeness and persuasiveness.
    \item For rating prediction task, the results of different methods on all datasets are shown in Table~\ref{tab:acc}. Based on the results, we can easily find that SVD++ outperforms PMF in terms of RMSE and MAE metrics, demonstrating its superior ability to capture user preferences through incorporating additional implicit feedback. 
    With the help of the multi-task learning mechanism for user recommendation explanation generation and rating prediction, both NRT and NETE have achieved better recommendation quality. For example, NRT and NETE show significant improvements over PMF and SVD++ in terms of the MAE metric on the Beauty and Sports datasets. This aligns with findings from previous work~\cite{DBLP:conf/acl/ChengWLZ0LL23,li2023personalized,li-etal-2021-personalized}. 
    Moreover, by applying our proposed HF4Rec framework to different backbones, we observe that both the holistic-based and multi-perspective-based methods outperform the base models in most cases. For example, the holistic-based models and multi-perspective-based methods on the Beauty dataset exceed the base model by an average of 1.03\% (\emph{w.r.t.} RMSE) and 1.48\% (\emph{w.r.t.} MAE), and 2.41\% (\emph{w.r.t.} RMSE) and 2.55\% (\emph{w.r.t.} MAE), respectively.  These findings indicate that, compared to traditional static explanation-text fitting supervision mechanisms, our methods leverage the powerful role-playing and semantic understanding capabilities of large language models. This allows for real-time feedback on the explanations generated during the model optimization process, and through a reinforcement learning-based reward maximization objective, our methods achieve better recommendation performance.
\end{itemize}

\subsection{Ablation Studies (RQ2)}

To validate the effectiveness of each component, we design five variants of our HF4Rec method as follows:
\begin{itemize}
    \item \textit{w/o}-DIS: This variant disables difficulty-aware sampling strategy and instead randomly sample data from unobserved data.
    \item \textit{w/o}-ICL:  This variant removes in-context-learning technique, retaining only user and target item information.
    \item \textit{w/o}-CSE: Instead of using customized reward scoring examples, this variant directly employs human-written prompt prototypes.
    \item \textit{w/o}-RAE: This variant deactivates retrieval-enhanced user information extraction and instead randomly chooses $K$ interaction items from the user’s history.
    \item \textit{w/o}-RB: This variant drops the advantage function, directly using reward values to update the policy network.
\end{itemize}
We conduct ablation studies on the Beauty and Yelp datasets. The metrics of BERTScore-F1 and subjective evaluation \emph{w.r.t.} informativeness and persuasiveness are adopted for experimental assessment.

From the results in Figure~\ref{fig:ablation_study_on_beauty} and~\ref{fig:ablation_study_on_yelp}, we can observe that removing any component leads to a decrease in performance across various evaluation metrics. Specifically, removing CSE and RAE consistently results in a significant drop in both BERTScore-F1 and informativeness, highlighting the importance of optimization strategies focused on difficulty-aware samples and retrieval-based user information for enhancing overall explanation personalization and informativeness. Notably, we find that removing the In-Context-Learning (ICL) technique causes a severe performance decline across all metrics, indicating that few-shot learning is crucial for enabling large language models to adapt effectively to recommendation explanation generation and reward assessment. These findings emphasize the contributions of individual components in enhancing the model's performance and confirm that the full model, with all components intact, yields the most optimal quality of generated recommendation explanations.

\subsection{Feedback Consistency Between Human and LLM (RQ3)}\label{sec:consistency}
\begin{table}[htbp]
    \centering
    \caption{Performance comparison of explanation quality (\emph{w.r.t.} Informativeness and Persuasiveness) between Human and Large Language Model (LLM) on Beauty and Yelp datasets. The best performance results are denoted in \textbf{bold} fonts.}
    \label{tab:human_vs_llm}
    \renewcommand{\arraystretch}{1.15}
    \resizebox{\textwidth}{!}{
    \begin{tabular}{*{5}{c}|*{4}{c}} 
    \hline\hline
    & \multicolumn{4}{c|}{Beauty} & \multicolumn{4}{c}{Yelp} \\
     & Info.(Human) & Persv.(LLM) & Info.(Human) & Persv.(LLM) & Info.(Human) & Persv.(LLM) & Info.(Human) & Persv.(LLM) \\
     \toprule
     PETER    & 2.033±0.6046          & 1.973±0.6923          & 1.667±0.5375          & 1.833±0.6368          & 2.060±0.5800          & 2.053±0.5864          & 1.900±0.3958          & 1.733±0.4422          \\
PETER-H  & {2.113±0.7168}    & { 2.180±0.6640}    & { 1.800±0.5416}    & { 2.033±0.7063}    & { 2.093±0.6148}    & { 2.127±0.6250}    & { 2.000±0.4472}    & { 1.900±0.5972}    \\
PETER-M  & \textbf{2.300±0.7371} & \textbf{2.407±0.6229} & \textbf{2.100±0.5385} & \textbf{2.200±0.6000} & \textbf{2.507±0.5859} & \textbf{2.580±0.5689} & \textbf{2.100±0.3000} & \textbf{2.200±0.4000} \\
\hline
PEPLER   & 2.080±0.7350          & 2.100±0.7895          & 1.767±0.4955          & 1.933±0.6289          & 2.047±0.6667          & 2.213±0.6692          & 1.700±0.5260          & 1.867±0.5617          \\
PEPLER-H & { 2.140±0.7486}    & { 2.187±0.7338}    & { 1.833±0.6368}    & { 2.267±0.6289}    & { 2.247±0.6211}    & { 2.320±0.5694}    & { 2.100±0.3958}    & { 2.133±0.5617}    \\
PEPLER-M & \textbf{2.233±0.6968} & \textbf{2.407±0.5899} & \textbf{2.100±0.3958} & \textbf{2.367±0.4819} & \textbf{2.420±0.6030} & \textbf{2.520±0.5381} & \textbf{2.267±0.4422} & \textbf{2.367±0.6574} \\
\hline
ERRA     & 1.753±0.7205          & 2.060±0.6349          & 1.667±0.4714          & 2.000±0.5164          & 2.040±0.6621          & 2.153±0.6903          & 1.700±0.5859          & 1.900±0.5972          \\
ERRA-H   & { 1.900±0.6608}    & { 2.167±0.6046}    & { 1.933±0.5735}    & { 2.200±0.4761}    & { 2.313±0.6233}    & { 2.400±0.5657}    & { 2.200±0.5416}    & { 2.067±0.4422}    \\
ERRA-M   & \textbf{2.160±0.5547} & \textbf{2.333±0.5497} & \textbf{2.200±0.4761} & \textbf{2.533±0.5617} & \textbf{2.327±0.5354} & \textbf{2.440±0.5713} & \textbf{2.400±0.5538} & \textbf{2.333±0.5375}\\
    \hline\hline
    \end{tabular}
}
\end{table}

\begin{table}[htbp]
    \centering
    \caption{Average Spearman's correlation between Human-to-Human and Human-to-LLM on Beauty and Yelp datasets. Note that a higher Spearman's correlation indicates a more similar rank between the two variables (\emph{i.e.}, Human vs. Human or Human vs. LLM).}
    \label{tab:spearman_correlation}
    \renewcommand{\arraystretch}{1.15}
    \resizebox{\textwidth}{!}{
    \begin{tabular}{*{4}{c}|*{2}{c}|*{2}{c}} 
    \hline\hline
          &                 & \multicolumn{2}{c|}{PETER-based methods}       & \multicolumn{2}{c|}{PEPLER-based methods}      & \multicolumn{2}{c}{ERRA-based methods}        \\
   &                 & Human vs. Human & Human vs. LLM & Human vs. Human & Human vs. LLM & Human vs. Human & Human vs. LLM \\
   \toprule
\multirow{2}{*}{Beauty} & Informativeness & 0.3731          & 0.3425        & 0.1883          & 0.2754        & 0.7215          & 0.4509        \\
                        & Persuasiveness  & 0.4695          & 0.4830        & 0.1981          & 0.3652        & 0.4456          & 0.5476        \\
    \hline
\multirow{2}{*}{Yelp}   & Informativeness & 0.4800          & 0.5155        & 0.4502          & 0.3120        & 0.4338          & 0.3265        \\
                        & Persuasiveness  & 0.5457          & 0.4345        & 0.4500          & 0.4447        & 0.2933          & 0.4517 \\
    \hline\hline
    \end{tabular}
}
\end{table}

To verify the rationality of using LLMs to simulate real-human participation for providing rewards and experimental evaluations, we conduct extensive experiments on the Beauty and Yelp datasets and adopt the subjective metrics. Specifically, we randomly select 30 samples from each test dataset and employ 10 university student annotators from diverse backgrounds to evaluate three methods (\emph{i.e.}, X, X-H, and X-M) based on different base models (\emph{i.e.}, PETER, PEPLER, and ERRA). For the evaluation of informativeness and persuasiveness in the generated recommendation explanations, we adopt the same definitions provided in Prompt~\ref{prompt:reward_prompt} to guide annotators in assessing these two subjective metrics. The informativeness and persuasiveness scores from all annotators are averaged as the final results. 
The experimental results are shown in Table~\ref{tab:human_vs_llm}, it is evident to observe that there are high consistency performances in scoring trends across different datasets and quality perspectives between humans and LLMs, with the ranking being $\text{X-M} > \text{X-H} > \text{X}$. This indicates that the large language models could accurately differentiate the quality of generated explanations just as humans do.

\begin{figure}[t]
    \centering
    \includegraphics[width=\textwidth]{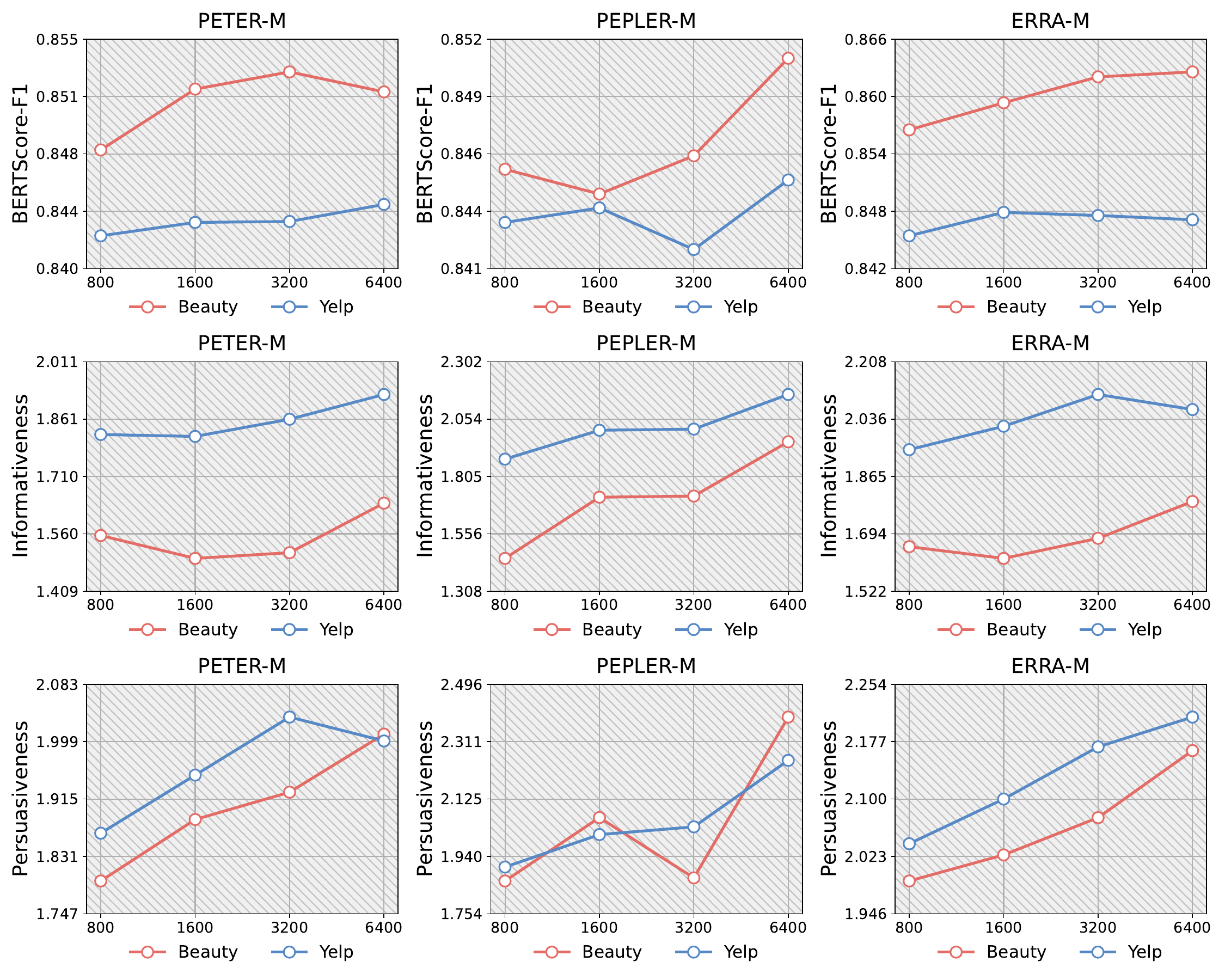}
    \caption{Trends in model performance \emph{w.r.t.} data volume on Beauty and Yelp datasets.}
    \label{fig:datasize_trend}
\end{figure}

Further, we introduce the Spearman correlation coefficient~\cite{zar2005spearman} to quantitatively measure the consistency in ranking tendencies both among human annotators and between human annotators and LLMs. Higher values indicate stronger alignment in their relative assessments. Additionally, future work may consider employing Cohen's Kappa to explicitly evaluate annotation agreement, capturing the exact match consistency between annotators and LLMs.
We calculate the average Spearman correlation both among annotators, and between annotators and LLM. In specific, as shown in Table~\ref{tab:spearman_correlation}, 
the average Spearman correlation coefficients between human annotators for the informativeness and persuasiveness metrics are 0.4411 and 0.4004, respectively. This indicates a moderate positive correlation among annotators, confirming the reasonableness and effectiveness of our experimental setup. Additionally, by comparing the consistency results between Human-Human and Human-LLM evaluations, we observe that the LLM achieves evaluation consistency comparable to humans. Remarkably, in some cases (6 out of 12), the LLM even surpasses human annotators. These results suggest that the LLM we designed demonstrates strong human-like capability in subjective evaluations, thereby enhancing user satisfaction with model-generated explanations.
Actually, this result is not surprising, because LLMs are trained on large-scale knowledge corpora and, coupled with our carefully designed customized reward prompts, effectively leverage their emergent language comprehension and logical reasoning capabilities, thereby achieving the excellent human-like performance.

\subsection{Parameter Analysis (RQ4)}

\begin{figure}[t]
    \centering
    \includegraphics[width=0.9\textwidth]{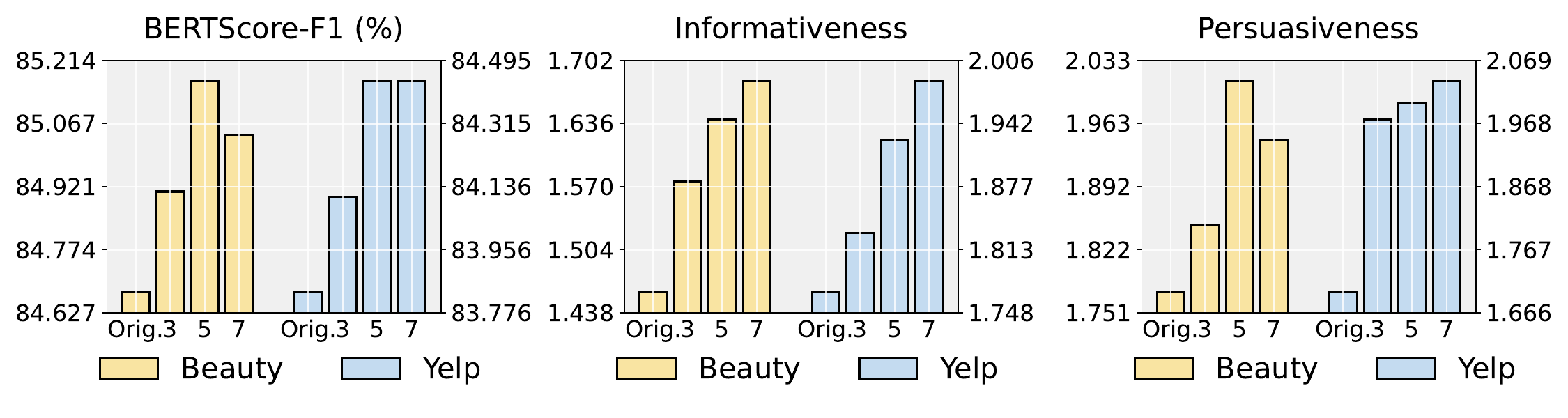}
    \caption{Trends in model performance \emph{w.r.t.} exploration number of generated explanations on Beauty and Yelp datasets.}
    \label{fig:expnum_trend}
\end{figure}

\subsubsection{Trends in model performance \emph{w.r.t.} training data volume}
Traditional explainable recommendation models usually follow the supervised learning paradigm, which are likely limited by existing observed interactions and user reviews. This problem can be alleviated by our method because the proposed optimization framework utilizes augmented human-like feedback for providing valuable learning signals to explored user behaviors. We change the data volume in the range of $\{800,1600,3200,6400\}$ to illustrate the trend in model performance. The evaluation results on the Beauty and Yelp datasets are shown in Figure~\ref{fig:datasize_trend}. As we can see, the overall trends for different evaluation metrics significantly improve when the data volume increases. This indicates that the HF4Rec framework is able to make better use of potential unobserved data to enhance the model’s generalization performance.

\subsubsection{Trends in model performance \emph{w.r.t.} exploration number of generated explanations}
In our method, we explore $J$ explanations for each user-item pair based on the text conditional probability generation mechanism of the language model. Then, the average reward value of these $J$ explanations is adopted as a baseline and replace the original reward in the policy network optimization process. The number of generated explanations impacts the performance of the explainable recommendation task. To investigate this, we vary the number of generated explanations across $\{3, 5, 7\}$. The experimental results, shown in Figure~\ref{fig:expnum_trend}, indicate that as the number increases, both objective and subjective metrics generally improve. This suggests that the introduction of the advantage function provides a clearer optimization direction for model parameter updates. Additionally, it helps the model quickly grasp the general patterns of high-quality explanations, leading to explanations that are more persuasive, informative, and accurate.

\subsection{Further Analysis}
Next, we continue to study whether HF4Rec works well in more detailed analysis.

\subsubsection{Efficiency of Pareto Optimization (RQ5)}

\begin{table}[t]
    \centering
    \caption{Performance comparison of different multi-perspective optimization approaches on the Beauty and Yelp datasets. `Orig.' denotes the original method, while `Avg.' refers to the optimization method using average weighted multi-perspective losses.}
    \label{tab:pareto}
    \renewcommand{\arraystretch}{1.15}
    \begin{tabular}{*{5}{c}|*{3}{c}|*{3}{c}} 
    \hline\hline
     & & \multicolumn{3}{c|}{PETER-based} & \multicolumn{3}{c|}{PEPLER-based} & \multicolumn{3}{c}{ERRA-based}\\
     & & Orig. & Avg. & Ours & Orig. & Avg. & Ours & Orig. & Avg. & Ours\\
     \toprule
          \multirow{2}{*}{Beauty} & Info. & 1.460 & 1.470 & \textbf{1.640} & 1.630 & 1.745 & \textbf{1.955} & 1.595 & 1.580 & \textbf{1.790}\\
        & Persv. & 1.775 & 1.840 & \textbf{2.010} & 1.830 & 2.095 & \textbf{2.390} & 1.915 & 2.050 & \textbf{2.165}\\
        \hline
        \multirow{2}{*}{Yelp} & Info. & 1.770 & 1.880 & \textbf{1.925} & 1.795 & 2.080 & \textbf{2.160} & 1.925 & 1.980 & \textbf{2.065}\\
        & Persv. & 1.700 & 1.965 & \textbf{2.000} & 1.800 & 2.205 & \textbf{2.250} & 1.875 & 2.100 & \textbf{2.210}\\
    \hline\hline
    \end{tabular}
\end{table}

\begin{figure}[t]
    \centering
    \includegraphics[width=\textwidth]{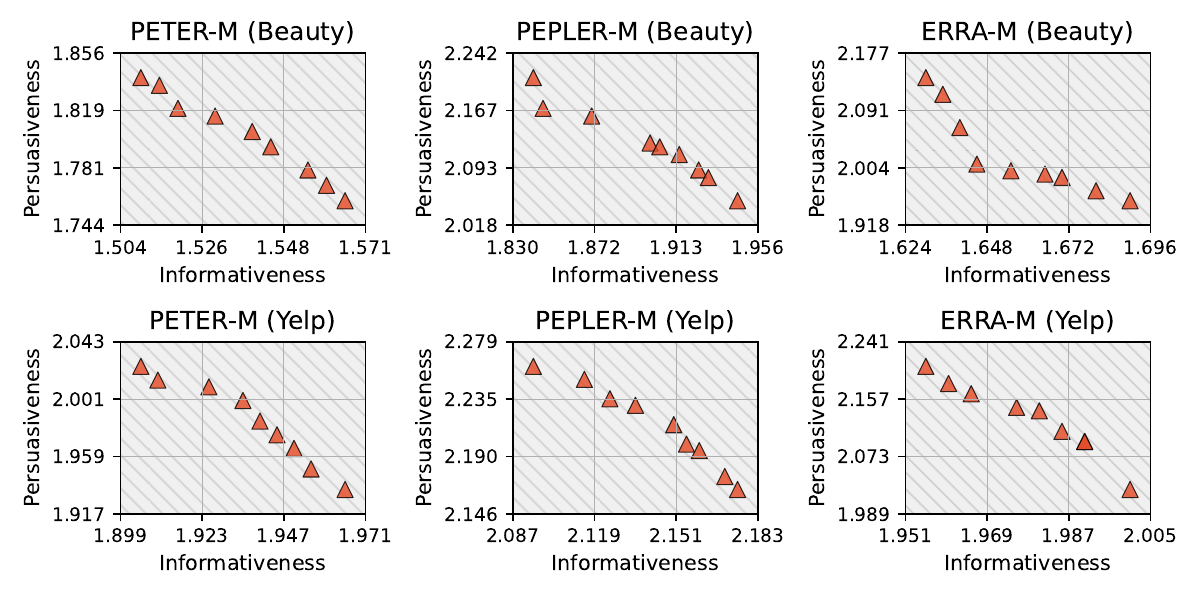}
    \vspace{-0.8cm}
    \caption{Impact analysis of multi-perspective prior preference vector on informativeness and persuasiveness.}
    \label{fig:pareto_front}
\end{figure}

To validate the efficacy of Pareto optimization in our HF4Rec framework, we compare it with the average weighted method on the Beauty and Yelp datasets, denoted as `Avg.'. As shown in Table~\ref{tab:pareto}, `Avg.' suffers from significant performance drops. We speculate an important reason is that the `Avg.' method's inability to adaptively adjust the weights of different objectives during the model training process. This may lead to the dominance of certain objective in the optimization process, potentially degrading the overall quality of the explanations. In contrast, our principled method can strike a balance between different perspectives of explanation quality. Further, to verify whether our optimization approach can attain the Pareto front, we fix $\beta_1 + \beta_2 = 0.8$, and incrementally adjust $\beta_1$ from 0 to 0.8 in intervals of 0.1. The experimental results are shown in Figure~\ref{fig:pareto_front}, it can be observed that the multi-perspective Pareto optimization approach successfully balances the trade-off between informativeness and persuasiveness, leading to the satisfactory performance on Pareto front, that is, no objective can be improved without sacrificing at least one other. This clearly showcases the effectiveness and efficiency of Pareto optimization strategy.

\subsubsection{Robustness Evaluation (RQ6)}

\begin{figure}[t]
    \centering
    \includegraphics[width=\textwidth]{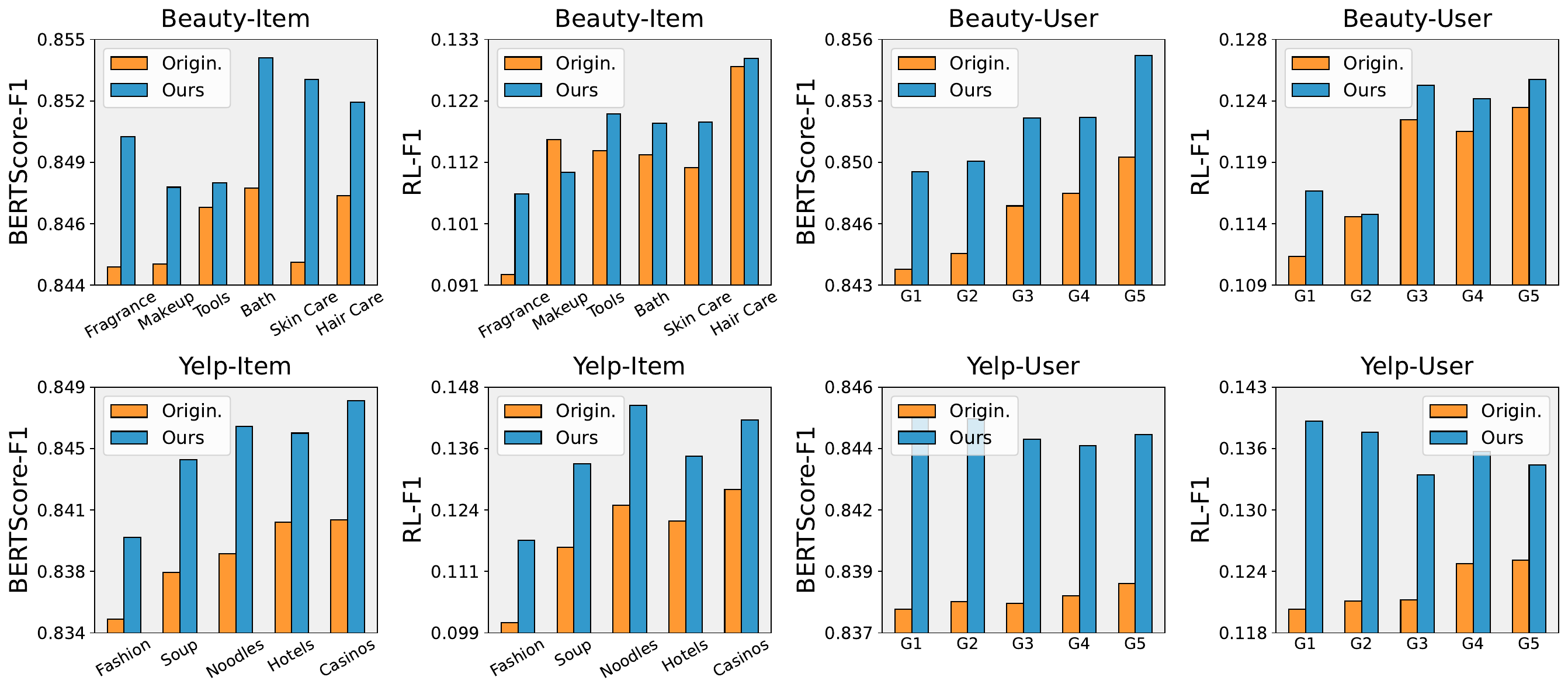}
    \caption{Robustness analysis of data sparsity.}
    \label{fig:sparse}
\end{figure}

To further validate that HF4Rec can improve the model robustness, we compare the performance of the original base model and HF4Rec with various data sparsity levels. Specifically, for items, we randomly select five item categories and sort them in ascending order based on the average number of interactions. For users, we divided them into five groups based on interaction frequency, while keeping the total number of interactions within each group constant. For representation convenience, we denote five user groups from lowest to highest activity as G1, G2, G3, G4, and G5. The detailed comparison results on PETER base model are displayed in Figure~\ref{fig:sparse}. From the results, we can find that HF4Rec beats the original models in most user and item groups. Moreover, as the number of interaction decreases, HF4Rec demonstrates clearer advantage in scenarios where items and users interact sparsely. This finding highlights the significance of incorporating human feedback to fine-tune models by utilizing augmented data obtained through difficulty-aware sampling strategy.

\subsubsection{Case Study}

\begin{figure}[t]
    \centering
    \includegraphics[width=\textwidth]{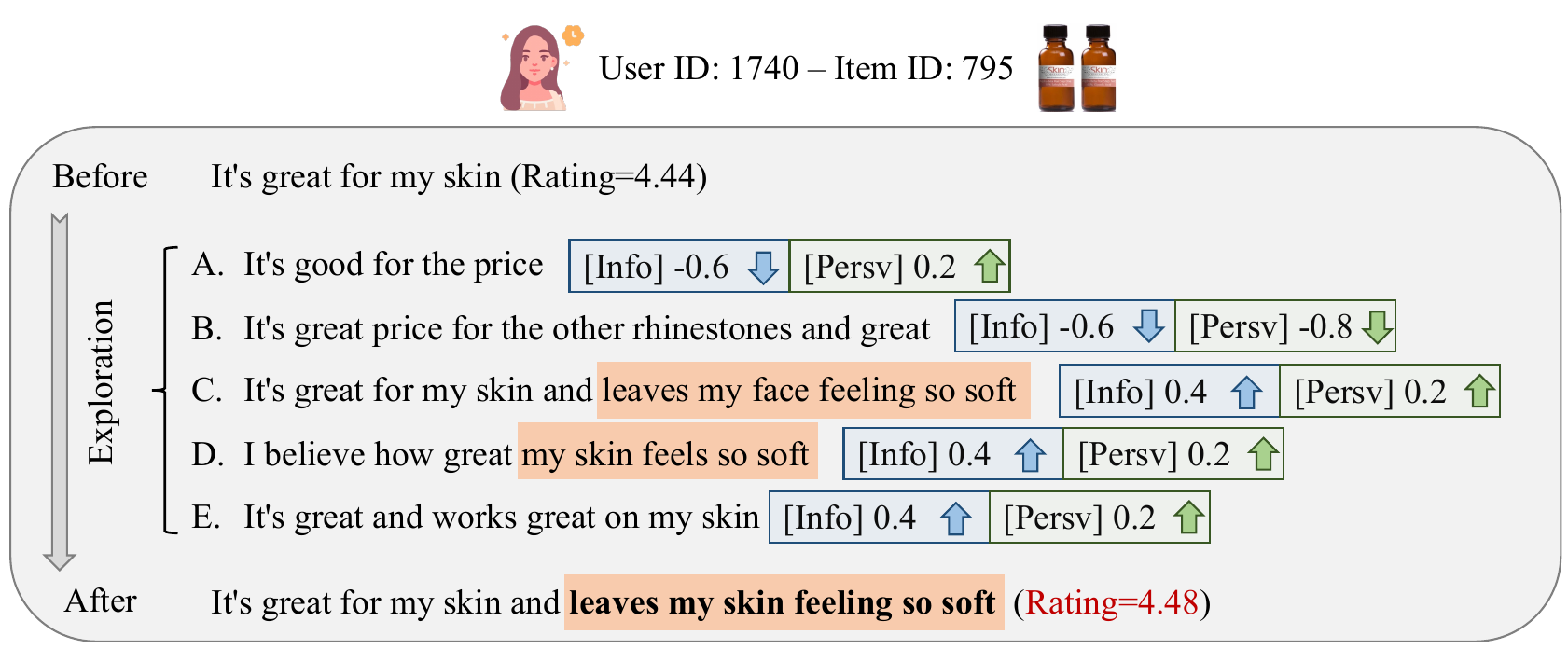}
    \caption{Case Study.}
    \label{fig:case}
\end{figure}

Figure~\ref{fig:case} presents a real instance from the Beauty dataset, illustrating the intricate details of internal optimization within the HF4Rec framework. Specifically, we randomly sample an unobserved interaction between user \#1740 and item \#795, a popular skin product. Before optimization, the model predict the recommendation explanation as ``It's great for my skin''. Subsequently, the HF4Rec framework explores five different explanations for this user-item pair. The LLM then assigns reward values for informativeness and persuasiveness to each explanation (as depicted in the Figure~\ref{fig:case} with BLEU blocks for informativeness and green blocks for persuasiveness), and these rewards are normalized using the designed advantage function in Section~\ref{sec:adv}. For example, the sentence `B' does not provide desired informativeness or persuasiveness, resulting in negative advantage values. In contrast, the user experience with the skin product are more crucial for user \#1740, leading to higher scores for explanations C, D, and E.

More importantly, the advantage values of different explored explanation provide explicit directions for parameter updating, that is, the model can be optimized towards explanations with positive advantages and away from those with negative ones. After optimization, the network effectively incorporates words with greater informativeness and persuasiveness from positively advantaged explanations (as highlighted in the Figure~\ref{fig:case}). This demonstrates HF4Rec's capability to generate valuable human-like feedback for unobserved data. Furthermore, as the model’s performance of explanation generation enhances, it also facilitates improving the accuracy of the rating predictions, as evidenced by the increase in the score from 4.44 to 4.48.

\section{Related Work}
In the era of information overload, recommender systems (RS) play a pivotal role in filtering and achieving personalized user experiences~\cite{wang2023recagent,10.1145/3589334.3645537,10.1145/3442381.3450039,10.1145/3459637.3482016,10.1145/3569423,10.1145/3580305.3599519}. However, with the deep recommendation algorithms growing more sophisticated, there is a rising call for transparency and comprehensibility in their decision-making mechanism. Explainable recommender systems (ERS) meet this requirement by providing insights into the reasoning behind recommendation outputs, thereby boosting user engagement and satisfaction. 

In recent years, a variety of explainable recommendation models (ERM) have rapidly emerged, including techniques such as reasoning rule mining~\cite{chen2021neural,shi2020neural,zhu-etal-2021-faithfully,gao2019explainable,Balog_Radlinski_Arakelyan_2019,peake2018explanation}, predefined explanation template filling~\cite{Li_Chen_Dong_2021,tan2021counterfactual,zhang2014explicit,wang2018explainable}, and knowledge graphs path generation~\cite{ai2018learning,fu2020fairness,xian2019reinforcement,xian2020cafe,chu2024llm,wang2019explainable,wang2024reinforced}. However, these methods usually require extensive manual labor and struggle to provide personalized, human-like, and flexible recommendation explanations. With the rapid advancement of natural language processing technologies, explainable recommendation methods based on natural language text generation have become increasingly popular~\cite{li2023personalized}. These approaches leverage state-of-the-art language models to provide explanations that are not only flexible but also easily comprehensible to users, marking a significant shift towards more user-centered recommendation systems. For example, the attribute-to-sequence (Att2Seq)~\cite{dong2017learning} model employs an LSTM-based attention-enhanced approach to predict product reviews from given attributes of users, products, and ratings. Similarly, the Neural Rating and Tips generation (NRT)~\cite{li2017neural} utilizes a GRU-based framework to jointly model rating predictions and tips generation, translating the latent representations of users and items into sentences that mimic human experiences and sentiments. The Neural Template (NETE)~\cite{li2020generate} explanation generation framework learns sentence templates from available data, producing template-controlled sentences that enhance the expressiveness of explanations. To improve the personalization of predicted explanations, PETER~\cite{li-etal-2021-personalized} leverages user and item IDs to predict target explanations, endowing IDs with linguistic meaning, and pioneers the application of the Transformer-based architecture~\cite{vaswani2017attention} to personalized natural language generation tasks. Furthermore, by utilizing retrieval-enhanced methods, ERRA~\cite{DBLP:conf/acl/ChengWLZ0LL23} generates more informative and accurate explanations by incorporating additional information from the training data and using aspect-enhancement components to better model user preferences with more relevant details. Since model-generated explanations may suffer from issues such as repetition and factual inaccuracies, PRAG~\cite{xie2023factual} proposes integrating a retrieval mechanism with a generative model to produce informative and factually accurate recommendation explanations by leveraging personalized information from users and items.
Moreover, as large language models (LLMs) have gained extensive knowledge and capabilities acquired from massive corpora~\cite{DBLP:journals/tmlr/WeiTBRZBYBZMCHVLDF22,zhao2023survey,wei2022chain}, they are increasing proving their effectiveness in explainable recommendation tasks. The methods like PEPLER~\cite{li2023personalized}, LLM2ER-EQR~\cite{yang2024fine}, KnowRec~\cite{colas-etal-2023-knowledge}, and LLMHG~\cite{chu2024llm} enhance personalized explainable models using GPT-series~\cite{10.5555/3495724.3495883}, BERT~\cite{devlin-etal-2019-bert}, BART~\cite{lewis-etal-2020-bart}, etc., showcasing their unprecedented potential.

However, existing models mostly follow a supervised learning approach, fitting sparse textual data and lacking the capability to provide valuable feedback signals for potentially better or worse explanation outputs. Additionally, the evaluation of explanation quality is essentially multi-perspective, with potential conflicts among different aspects. Previous explainable recommendation methods overlook how to balance multiple quality enhancement objectives, resulting in suboptimal performance. Different from these traditional approaches, this paper innovatively proposes a human-like feedback-driven optimization framework that reconsiders the human-centered explainable recommendation task. we employ large language models for human-like assessments of explored explanation predictions, and introduce Pareto optimization to improve multi-perspective quality objectives.

\section{Conclusion}
In this paper, we introduce HF4Rec, a human-like feedback-driven optimization framework for enhancing the quality and robustness of explainable recommendation models. Leveraging Large Language Models (LLMs) as human simulators, our framework provides accurate reward estimations for exploration sets of generated explanations, which effectively provide the valuable learning signals of model parameter updating via the policy-gradient mechanism. Furthermore, we incorporate Pareto optimization to mitigate potential conflicts among different aspects of explanations, simultaneously maximizing multiple objective functions to improve multi-perspective explanation performance. Additionally, HF4Rec adopts an off-policy optimization pipeline with a replay buffer and importance sampling-based bias correction, significantly improving model generality. We conduct extensive experiments to demonstrate the HF4Rec's superiority in objective and subjective evaluations of explanation generation and recommendation accuracy, validating the reasonability using LLMs as proxies for human-like feedback and experimental evaluation.

This paper endeavors to integrate human feedback into the reinforcement learning optimization framework to enhance the multi-perspective quality of recommendation explanations. In this light, there are numerous promising directions for future research. Specifically, we can apply our optimization framework to larger explainable recommendation models, such as models augmented by multi-task oriented recommendation large models. Additionally, the methods proposed in this paper are not limited to explainable recommendation tasks but can also be extended to other human-centric scenarios, such as dialogue systems. By leveraging the substantial language comprehension and reasoning capabilities of LLMs, we can refine traditional model optimization methods to more closely align with actual human preferences.



\bibliographystyle{ACM-Reference-Format}
\bibliography{sample-base}


\begin{thebibliography}{87}


\ifx \showCODEN    \undefined \def \showCODEN     #1{\unskip}     \fi
\ifx \showDOI      \undefined \def \showDOI       #1{#1}\fi
\ifx \showISBNx    \undefined \def \showISBNx     #1{\unskip}     \fi
\ifx \showISBNxiii \undefined \def \showISBNxiii  #1{\unskip}     \fi
\ifx \showISSN     \undefined \def \showISSN      #1{\unskip}     \fi
\ifx \showLCCN     \undefined \def \showLCCN      #1{\unskip}     \fi
\ifx \shownote     \undefined \def \shownote      #1{#1}          \fi
\ifx \showarticletitle \undefined \def \showarticletitle #1{#1}   \fi
\ifx \showURL      \undefined \def \showURL       {\relax}        \fi
\providecommand\bibfield[2]{#2}
\providecommand\bibinfo[2]{#2}
\providecommand\natexlab[1]{#1}
\providecommand\showeprint[2][]{arXiv:#2}

\bibitem[Aher et~al\mbox{.}(2023)]%
        {aher2023using}
\bibfield{author}{\bibinfo{person}{Gati~V Aher}, \bibinfo{person}{Rosa~I Arriaga}, {and} \bibinfo{person}{Adam~Tauman Kalai}.} \bibinfo{year}{2023}\natexlab{}.
\newblock \showarticletitle{Using large language models to simulate multiple humans and replicate human subject studies}. In \bibinfo{booktitle}{\emph{International Conference on Machine Learning}}. PMLR, \bibinfo{pages}{337--371}.
\newblock


\bibitem[Ai et~al\mbox{.}(2018)]%
        {ai2018learning}
\bibfield{author}{\bibinfo{person}{Qingyao Ai}, \bibinfo{person}{Vahid Azizi}, \bibinfo{person}{Xu Chen}, {and} \bibinfo{person}{Yongfeng Zhang}.} \bibinfo{year}{2018}\natexlab{}.
\newblock \showarticletitle{Learning heterogeneous knowledge base embeddings for explainable recommendation}.
\newblock \bibinfo{journal}{\emph{Algorithms}} \bibinfo{volume}{11}, \bibinfo{number}{9} (\bibinfo{year}{2018}), \bibinfo{pages}{137}.
\newblock


\bibitem[Balog et~al\mbox{.}(2019)]%
        {Balog_Radlinski_Arakelyan_2019}
\bibfield{author}{\bibinfo{person}{Krisztian Balog}, \bibinfo{person}{Filip Radlinski}, {and} \bibinfo{person}{Shushan Arakelyan}.} \bibinfo{year}{2019}\natexlab{}.
\newblock \showarticletitle{Transparent, Scrutable and Explainable User Models for Personalized Recommendation}. In \bibinfo{booktitle}{\emph{Proceedings of the 42nd International ACM SIGIR Conference on Research and Development in Information Retrieval}}.
\newblock


\bibitem[Brown et~al\mbox{.}(2020b)]%
        {brown2020language}
\bibfield{author}{\bibinfo{person}{Tom Brown}, \bibinfo{person}{Benjamin Mann}, \bibinfo{person}{Nick Ryder}, \bibinfo{person}{Melanie Subbiah}, \bibinfo{person}{Jared~D Kaplan}, \bibinfo{person}{Prafulla Dhariwal}, \bibinfo{person}{Arvind Neelakantan}, \bibinfo{person}{Pranav Shyam}, \bibinfo{person}{Girish Sastry}, \bibinfo{person}{Amanda Askell}, {et~al\mbox{.}}} \bibinfo{year}{2020}\natexlab{b}.
\newblock \showarticletitle{Language models are few-shot learners}.
\newblock \bibinfo{journal}{\emph{Advances in neural information processing systems}}  \bibinfo{volume}{33} (\bibinfo{year}{2020}), \bibinfo{pages}{1877--1901}.
\newblock


\bibitem[Brown et~al\mbox{.}(2020a)]%
        {10.5555/3495724.3495883}
\bibfield{author}{\bibinfo{person}{Tom~B. Brown}, \bibinfo{person}{Benjamin Mann}, \bibinfo{person}{Nick Ryder}, \bibinfo{person}{Melanie Subbiah}, \bibinfo{person}{Jared Kaplan}, \bibinfo{person}{Prafulla Dhariwal}, \bibinfo{person}{Arvind Neelakantan}, \bibinfo{person}{Pranav Shyam}, \bibinfo{person}{Girish Sastry}, \bibinfo{person}{Amanda Askell}, \bibinfo{person}{Sandhini Agarwal}, \bibinfo{person}{Ariel Herbert-Voss}, \bibinfo{person}{Gretchen Krueger}, \bibinfo{person}{Tom Henighan}, \bibinfo{person}{Rewon Child}, \bibinfo{person}{Aditya Ramesh}, \bibinfo{person}{Daniel~M. Ziegler}, \bibinfo{person}{Jeffrey Wu}, \bibinfo{person}{Clemens Winter}, \bibinfo{person}{Christopher Hesse}, \bibinfo{person}{Mark Chen}, \bibinfo{person}{Eric Sigler}, \bibinfo{person}{Mateusz Litwin}, \bibinfo{person}{Scott Gray}, \bibinfo{person}{Benjamin Chess}, \bibinfo{person}{Jack Clark}, \bibinfo{person}{Christopher Berner}, \bibinfo{person}{Sam McCandlish}, \bibinfo{person}{Alec Radford}, \bibinfo{person}{Ilya Sutskever},
  {and} \bibinfo{person}{Dario Amodei}.} \bibinfo{year}{2020}\natexlab{a}.
\newblock \showarticletitle{Language models are few-shot learners}. In \bibinfo{booktitle}{\emph{Proceedings of the 34th International Conference on Neural Information Processing Systems}} \emph{(\bibinfo{series}{NIPS '20})}. \bibinfo{publisher}{Curran Associates Inc.}, \bibinfo{address}{Red Hook, NY, USA}, \bibinfo{numpages}{25}~pages.
\newblock
\showISBNx{9781713829546}


\bibitem[Censor(1977)]%
        {censor1977pareto}
\bibfield{author}{\bibinfo{person}{Yair Censor}.} \bibinfo{year}{1977}\natexlab{}.
\newblock \showarticletitle{Pareto optimality in multiobjective problems}.
\newblock \bibinfo{journal}{\emph{Applied Mathematics and Optimization}} \bibinfo{volume}{4}, \bibinfo{number}{1} (\bibinfo{year}{1977}), \bibinfo{pages}{41--59}.
\newblock


\bibitem[Chen et~al\mbox{.}(2021b)]%
        {chen2021neural}
\bibfield{author}{\bibinfo{person}{Hanxiong Chen}, \bibinfo{person}{Shaoyun Shi}, \bibinfo{person}{Yunqi Li}, {and} \bibinfo{person}{Yongfeng Zhang}.} \bibinfo{year}{2021}\natexlab{b}.
\newblock \showarticletitle{Neural collaborative reasoning}. In \bibinfo{booktitle}{\emph{Proceedings of the Web Conference 2021}}. \bibinfo{pages}{1516--1527}.
\newblock


\bibitem[Chen et~al\mbox{.}(2019)]%
        {10.1145/3331184.3331254}
\bibfield{author}{\bibinfo{person}{Xu Chen}, \bibinfo{person}{Hanxiong Chen}, \bibinfo{person}{Hongteng Xu}, \bibinfo{person}{Yongfeng Zhang}, \bibinfo{person}{Yixin Cao}, \bibinfo{person}{Zheng Qin}, {and} \bibinfo{person}{Hongyuan Zha}.} \bibinfo{year}{2019}\natexlab{}.
\newblock \showarticletitle{Personalized Fashion Recommendation with Visual Explanations based on Multimodal Attention Network: Towards Visually Explainable Recommendation}. In \bibinfo{booktitle}{\emph{Proceedings of the 42nd International ACM SIGIR Conference on Research and Development in Information Retrieval}} \emph{(\bibinfo{series}{SIGIR'19})}. \bibinfo{publisher}{Association for Computing Machinery}, \bibinfo{address}{New York, NY, USA}, \bibinfo{pages}{765–774}.
\newblock
\showISBNx{9781450361729}
\urldef\tempurl%
\url{https://doi.org/10.1145/3331184.3331254}
\showDOI{\tempurl}


\bibitem[Chen et~al\mbox{.}(2021a)]%
        {chen2021reinforcement}
\bibfield{author}{\bibinfo{person}{Xu Chen}, \bibinfo{person}{Yali Du}, \bibinfo{person}{Long Xia}, {and} \bibinfo{person}{Jun Wang}.} \bibinfo{year}{2021}\natexlab{a}.
\newblock \showarticletitle{Reinforcement recommendation with user multi-aspect preference}. In \bibinfo{booktitle}{\emph{Proceedings of the Web Conference 2021}}. \bibinfo{pages}{425--435}.
\newblock


\bibitem[Cheng et~al\mbox{.}(2023)]%
        {DBLP:conf/acl/ChengWLZ0LL23}
\bibfield{author}{\bibinfo{person}{Hao Cheng}, \bibinfo{person}{Shuo Wang}, \bibinfo{person}{Wensheng Lu}, \bibinfo{person}{Wei Zhang}, \bibinfo{person}{Mingyang Zhou}, \bibinfo{person}{Kezhong Lu}, {and} \bibinfo{person}{Hao Liao}.} \bibinfo{year}{2023}\natexlab{}.
\newblock \showarticletitle{Explainable Recommendation with Personalized Review Retrieval and Aspect Learning}. In \bibinfo{booktitle}{\emph{Proceedings of the 61st Annual Meeting of the Association for Computational Linguistics (Volume 1: Long Papers), {ACL} 2023, Toronto, Canada, July 9-14, 2023}}. \bibinfo{publisher}{Association for Computational Linguistics}, \bibinfo{pages}{51--64}.
\newblock


\bibitem[Chu et~al\mbox{.}(2024)]%
        {chu2024llm}
\bibfield{author}{\bibinfo{person}{Zhixuan Chu}, \bibinfo{person}{Yan Wang}, \bibinfo{person}{Qing Cui}, \bibinfo{person}{Longfei Li}, \bibinfo{person}{Wenqing Chen}, \bibinfo{person}{Sheng Li}, \bibinfo{person}{Zhan Qin}, {and} \bibinfo{person}{Kui Ren}.} \bibinfo{year}{2024}\natexlab{}.
\newblock \showarticletitle{Llm-guided multi-view hypergraph learning for human-centric explainable recommendation}.
\newblock \bibinfo{journal}{\emph{arXiv preprint arXiv:2401.08217}} (\bibinfo{year}{2024}).
\newblock


\bibitem[Chung et~al\mbox{.}(2014)]%
        {69e088c8129341ac89810907fe6b1bfe}
\bibfield{author}{\bibinfo{person}{Junyoung Chung}, \bibinfo{person}{Caglar Gulcehre}, \bibinfo{person}{Kyunghyun Cho}, {and} \bibinfo{person}{Yoshua Bengio}.} \bibinfo{year}{2014}\natexlab{}.
\newblock \showarticletitle{Empirical evaluation of gated recurrent neural networks on sequence modeling}. In \bibinfo{booktitle}{\emph{NIPS 2014 Workshop on Deep Learning, December 2014}}.
\newblock


\bibitem[Colas et~al\mbox{.}(2023)]%
        {colas-etal-2023-knowledge}
\bibfield{author}{\bibinfo{person}{Anthony Colas}, \bibinfo{person}{Jun Araki}, \bibinfo{person}{Zhengyu Zhou}, \bibinfo{person}{Bingqing Wang}, {and} \bibinfo{person}{Zhe Feng}.} \bibinfo{year}{2023}\natexlab{}.
\newblock \showarticletitle{Knowledge-Grounded Natural Language Recommendation Explanation}. In \bibinfo{booktitle}{\emph{Proceedings of the 6th BlackboxNLP Workshop: Analyzing and Interpreting Neural Networks for NLP}}. \bibinfo{publisher}{Association for Computational Linguistics}, \bibinfo{address}{Singapore}, \bibinfo{pages}{1--15}.
\newblock


\bibitem[Degris et~al\mbox{.}(2012)]%
        {10.5555/3042573.3042600}
\bibfield{author}{\bibinfo{person}{Thomas Degris}, \bibinfo{person}{Martha White}, {and} \bibinfo{person}{Richard~S. Sutton}.} \bibinfo{year}{2012}\natexlab{}.
\newblock \showarticletitle{Off-policy actor-critic}. In \bibinfo{booktitle}{\emph{Proceedings of the 29th International Coference on International Conference on Machine Learning}} (Edinburgh, Scotland) \emph{(\bibinfo{series}{ICML'12})}. \bibinfo{publisher}{Omnipress}, \bibinfo{pages}{179–186}.
\newblock
\showISBNx{9781450312851}


\bibitem[Devlin et~al\mbox{.}(2019)]%
        {devlin-etal-2019-bert}
\bibfield{author}{\bibinfo{person}{Jacob Devlin}, \bibinfo{person}{Ming-Wei Chang}, \bibinfo{person}{Kenton Lee}, {and} \bibinfo{person}{Kristina Toutanova}.} \bibinfo{year}{2019}\natexlab{}.
\newblock \showarticletitle{{BERT}: Pre-training of Deep Bidirectional Transformers for Language Understanding}. In \bibinfo{booktitle}{\emph{Proceedings of the 2019 Conference of the North {A}merican Chapter of the Association for Computational Linguistics: Human Language Technologies, Volume 1 (Long and Short Papers)}}. \bibinfo{publisher}{Association for Computational Linguistics}, \bibinfo{address}{Minneapolis, Minnesota}, \bibinfo{pages}{4171--4186}.
\newblock


\bibitem[Dong et~al\mbox{.}(2017)]%
        {dong2017learning}
\bibfield{author}{\bibinfo{person}{Li Dong}, \bibinfo{person}{Shaohan Huang}, \bibinfo{person}{Furu Wei}, \bibinfo{person}{Mirella Lapata}, \bibinfo{person}{Ming Zhou}, {and} \bibinfo{person}{Ke Xu}.} \bibinfo{year}{2017}\natexlab{}.
\newblock \showarticletitle{Learning to generate product reviews from attributes}. In \bibinfo{booktitle}{\emph{Proceedings of the 15th Conference of the European Chapter of the Association for Computational Linguistics}}. \bibinfo{pages}{623--632}.
\newblock


\bibitem[Fu et~al\mbox{.}(2020)]%
        {fu2020fairness}
\bibfield{author}{\bibinfo{person}{Zuohui Fu}, \bibinfo{person}{Yikun Xian}, \bibinfo{person}{Ruoyuan Gao}, \bibinfo{person}{Jieyu Zhao}, \bibinfo{person}{Qiaoying Huang}, \bibinfo{person}{Yingqiang Ge}, \bibinfo{person}{Shuyuan Xu}, \bibinfo{person}{Shijie Geng}, \bibinfo{person}{Chirag Shah}, \bibinfo{person}{Yongfeng Zhang}, {et~al\mbox{.}}} \bibinfo{year}{2020}\natexlab{}.
\newblock \showarticletitle{Fairness-aware explainable recommendation over knowledge graphs}. In \bibinfo{booktitle}{\emph{Proceedings of the 43rd international ACM SIGIR conference on research and development in information retrieval}}. \bibinfo{pages}{69--78}.
\newblock


\bibitem[Gao et~al\mbox{.}(2019)]%
        {gao2019explainable}
\bibfield{author}{\bibinfo{person}{Jingyue Gao}, \bibinfo{person}{Xiting Wang}, \bibinfo{person}{Yasha Wang}, {and} \bibinfo{person}{Xing Xie}.} \bibinfo{year}{2019}\natexlab{}.
\newblock \showarticletitle{Explainable recommendation through attentive multi-view learning}. In \bibinfo{booktitle}{\emph{Proceedings of the AAAI Conference on Artificial Intelligence}}, Vol.~\bibinfo{volume}{33}. \bibinfo{pages}{3622--3629}.
\newblock


\bibitem[Geng et~al\mbox{.}(2022)]%
        {geng2022recommendation}
\bibfield{author}{\bibinfo{person}{Shijie Geng}, \bibinfo{person}{Shuchang Liu}, \bibinfo{person}{Zuohui Fu}, \bibinfo{person}{Yingqiang Ge}, {and} \bibinfo{person}{Yongfeng Zhang}.} \bibinfo{year}{2022}\natexlab{}.
\newblock \showarticletitle{Recommendation as language processing (rlp): A unified pretrain, personalized prompt \& predict paradigm (p5)}. In \bibinfo{booktitle}{\emph{Proceedings of the 16th ACM Conference on Recommender Systems}}. \bibinfo{pages}{299--315}.
\newblock


\bibitem[Hartigan and Wong(1979)]%
        {hartigan1979algorithm}
\bibfield{author}{\bibinfo{person}{John~A Hartigan} {and} \bibinfo{person}{Manchek~A Wong}.} \bibinfo{year}{1979}\natexlab{}.
\newblock \showarticletitle{Algorithm AS 136: A k-means clustering algorithm}.
\newblock \bibinfo{journal}{\emph{Journal of the royal statistical society. series c (applied statistics)}} \bibinfo{volume}{28}, \bibinfo{number}{1} (\bibinfo{year}{1979}), \bibinfo{pages}{100--108}.
\newblock


\bibitem[Hochman and Rodgers(1969)]%
        {hochman1969pareto}
\bibfield{author}{\bibinfo{person}{Harold~M Hochman} {and} \bibinfo{person}{James~D Rodgers}.} \bibinfo{year}{1969}\natexlab{}.
\newblock \showarticletitle{Pareto optimal redistribution}.
\newblock \bibinfo{journal}{\emph{The American economic review}} \bibinfo{volume}{59}, \bibinfo{number}{4} (\bibinfo{year}{1969}), \bibinfo{pages}{542--557}.
\newblock


\bibitem[Hochreiter and Schmidhuber(1997)]%
        {hochreiter1997long}
\bibfield{author}{\bibinfo{person}{Sepp Hochreiter} {and} \bibinfo{person}{J{\"u}rgen Schmidhuber}.} \bibinfo{year}{1997}\natexlab{}.
\newblock \showarticletitle{Long short-term memory}.
\newblock \bibinfo{journal}{\emph{Neural computation}} \bibinfo{volume}{9}, \bibinfo{number}{8} (\bibinfo{year}{1997}), \bibinfo{pages}{1735--1780}.
\newblock


\bibitem[Huang and Chang(2022)]%
        {huang2022towards}
\bibfield{author}{\bibinfo{person}{Jie Huang} {and} \bibinfo{person}{Kevin Chen-Chuan Chang}.} \bibinfo{year}{2022}\natexlab{}.
\newblock \showarticletitle{Towards reasoning in large language models: A survey}.
\newblock \bibinfo{journal}{\emph{arXiv preprint arXiv:2212.10403}} (\bibinfo{year}{2022}).
\newblock


\bibitem[Kingma and Ba(2014)]%
        {kingma2014adam}
\bibfield{author}{\bibinfo{person}{Diederik~P Kingma} {and} \bibinfo{person}{Jimmy Ba}.} \bibinfo{year}{2014}\natexlab{}.
\newblock \showarticletitle{Adam: A method for stochastic optimization}.
\newblock \bibinfo{journal}{\emph{arXiv preprint arXiv:1412.6980}} (\bibinfo{year}{2014}).
\newblock


\bibitem[Konda and Tsitsiklis(1999)]%
        {konda1999actor}
\bibfield{author}{\bibinfo{person}{Vijay Konda} {and} \bibinfo{person}{John Tsitsiklis}.} \bibinfo{year}{1999}\natexlab{}.
\newblock \showarticletitle{Actor-critic algorithms}.
\newblock \bibinfo{journal}{\emph{Advances in neural information processing systems}}  \bibinfo{volume}{12} (\bibinfo{year}{1999}).
\newblock


\bibitem[Koren(2010)]%
        {koren2010factor}
\bibfield{author}{\bibinfo{person}{Yehuda Koren}.} \bibinfo{year}{2010}\natexlab{}.
\newblock \showarticletitle{Factor in the neighbors: Scalable and accurate collaborative filtering}.
\newblock \bibinfo{journal}{\emph{ACM Transactions on Knowledge Discovery from Data (TKDD)}} \bibinfo{volume}{4}, \bibinfo{number}{1} (\bibinfo{year}{2010}), \bibinfo{pages}{1--24}.
\newblock


\bibitem[Lewis et~al\mbox{.}(2020)]%
        {lewis-etal-2020-bart}
\bibfield{author}{\bibinfo{person}{Mike Lewis}, \bibinfo{person}{Yinhan Liu}, \bibinfo{person}{Naman Goyal}, \bibinfo{person}{Marjan Ghazvininejad}, \bibinfo{person}{Abdelrahman Mohamed}, \bibinfo{person}{Omer Levy}, \bibinfo{person}{Veselin Stoyanov}, {and} \bibinfo{person}{Luke Zettlemoyer}.} \bibinfo{year}{2020}\natexlab{}.
\newblock \showarticletitle{{BART}: Denoising Sequence-to-Sequence Pre-training for Natural Language Generation, Translation, and Comprehension}. In \bibinfo{booktitle}{\emph{Proceedings of the 58th Annual Meeting of the Association for Computational Linguistics}}. \bibinfo{publisher}{Association for Computational Linguistics}, \bibinfo{address}{Online}, \bibinfo{pages}{7871--7880}.
\newblock


\bibitem[Li et~al\mbox{.}(2023a)]%
        {li2023ucepic}
\bibfield{author}{\bibinfo{person}{Jiacheng Li}, \bibinfo{person}{Zhankui He}, \bibinfo{person}{Jingbo Shang}, {and} \bibinfo{person}{Julian McAuley}.} \bibinfo{year}{2023}\natexlab{a}.
\newblock \showarticletitle{Ucepic: Unifying aspect planning and lexical constraints for generating explanations in recommendation}. In \bibinfo{booktitle}{\emph{Proceedings of the 29th ACM SIGKDD Conference on Knowledge Discovery and Data Mining}}. \bibinfo{pages}{1248--1257}.
\newblock


\bibitem[Li et~al\mbox{.}(2023b)]%
        {10.1145/3580305.3599519}
\bibfield{author}{\bibinfo{person}{Jiacheng Li}, \bibinfo{person}{Ming Wang}, \bibinfo{person}{Jin Li}, \bibinfo{person}{Jinmiao Fu}, \bibinfo{person}{Xin Shen}, \bibinfo{person}{Jingbo Shang}, {and} \bibinfo{person}{Julian McAuley}.} \bibinfo{year}{2023}\natexlab{b}.
\newblock \showarticletitle{Text Is All You Need: Learning Language Representations for Sequential Recommendation}. In \bibinfo{booktitle}{\emph{Proceedings of the 29th ACM SIGKDD Conference on Knowledge Discovery and Data Mining}} \emph{(\bibinfo{series}{KDD '23})}. \bibinfo{publisher}{Association for Computing Machinery}, \bibinfo{address}{New York, NY, USA}, \bibinfo{pages}{1258–1267}.
\newblock
\showISBNx{9798400701030}


\bibitem[Li et~al\mbox{.}(2021a)]%
        {Li_Chen_Dong_2021}
\bibfield{author}{\bibinfo{person}{Lei Li}, \bibinfo{person}{Li Chen}, {and} \bibinfo{person}{Ruihai Dong}.} \bibinfo{year}{2021}\natexlab{a}.
\newblock \showarticletitle{CAESAR: context-aware explanation based on supervised attention for service recommendations}.
\newblock \bibinfo{journal}{\emph{Journal of Intelligent Information Systems}} (\bibinfo{date}{Aug} \bibinfo{year}{2021}), \bibinfo{pages}{147–170}.
\newblock


\bibitem[Li et~al\mbox{.}(2020)]%
        {li2020generate}
\bibfield{author}{\bibinfo{person}{Lei Li}, \bibinfo{person}{Yongfeng Zhang}, {and} \bibinfo{person}{Li Chen}.} \bibinfo{year}{2020}\natexlab{}.
\newblock \showarticletitle{Generate neural template explanations for recommendation}. In \bibinfo{booktitle}{\emph{Proceedings of the 29th ACM International Conference on Information \& Knowledge Management}}. \bibinfo{pages}{755--764}.
\newblock


\bibitem[Li et~al\mbox{.}(2021b)]%
        {li-etal-2021-personalized}
\bibfield{author}{\bibinfo{person}{Lei Li}, \bibinfo{person}{Yongfeng Zhang}, {and} \bibinfo{person}{Li Chen}.} \bibinfo{year}{2021}\natexlab{b}.
\newblock \showarticletitle{Personalized Transformer for Explainable Recommendation}. In \bibinfo{booktitle}{\emph{Proceedings of the 59th Annual Meeting of the Association for Computational Linguistics and the 11th International Joint Conference on Natural Language Processing}}. \bibinfo{publisher}{Association for Computational Linguistics}, \bibinfo{pages}{4947--4957}.
\newblock


\bibitem[Li et~al\mbox{.}(2023c)]%
        {10.1145/3569423}
\bibfield{author}{\bibinfo{person}{Lei Li}, \bibinfo{person}{Yongfeng Zhang}, {and} \bibinfo{person}{Li Chen}.} \bibinfo{year}{2023}\natexlab{c}.
\newblock \showarticletitle{On the Relationship between Explanation and Recommendation: Learning to Rank Explanations for Improved Performance}.
\newblock \bibinfo{journal}{\emph{ACM Trans. Intell. Syst. Technol.}} \bibinfo{volume}{14}, \bibinfo{number}{2}, Article \bibinfo{articleno}{21} (\bibinfo{date}{feb} \bibinfo{year}{2023}), \bibinfo{numpages}{24}~pages.
\newblock
\showISSN{2157-6904}


\bibitem[Li et~al\mbox{.}(2023d)]%
        {li2023personalized}
\bibfield{author}{\bibinfo{person}{Lei Li}, \bibinfo{person}{Yongfeng Zhang}, {and} \bibinfo{person}{Li Chen}.} \bibinfo{year}{2023}\natexlab{d}.
\newblock \showarticletitle{Personalized prompt learning for explainable recommendation}.
\newblock \bibinfo{journal}{\emph{ACM Transactions on Information Systems}} \bibinfo{volume}{41}, \bibinfo{number}{4} (\bibinfo{year}{2023}), \bibinfo{pages}{1--26}.
\newblock


\bibitem[Li et~al\mbox{.}(2023e)]%
        {li2023prompt}
\bibfield{author}{\bibinfo{person}{Lei Li}, \bibinfo{person}{Yongfeng Zhang}, {and} \bibinfo{person}{Li Chen}.} \bibinfo{year}{2023}\natexlab{e}.
\newblock \showarticletitle{Prompt distillation for efficient llm-based recommendation}. In \bibinfo{booktitle}{\emph{Proceedings of the 32nd ACM International Conference on Information and Knowledge Management}}. \bibinfo{pages}{1348--1357}.
\newblock


\bibitem[Li et~al\mbox{.}(2017)]%
        {li2017neural}
\bibfield{author}{\bibinfo{person}{Piji Li}, \bibinfo{person}{Zihao Wang}, \bibinfo{person}{Zhaochun Ren}, \bibinfo{person}{Lidong Bing}, {and} \bibinfo{person}{Wai Lam}.} \bibinfo{year}{2017}\natexlab{}.
\newblock \showarticletitle{Neural rating regression with abstractive tips generation for recommendation}. In \bibinfo{booktitle}{\emph{Proceedings of the 40th International ACM SIGIR conference on Research and Development in Information Retrieval}}. \bibinfo{pages}{345--354}.
\newblock


\bibitem[Lillicrap et~al\mbox{.}(2016)]%
        {DBLP:journals/corr/LillicrapHPHETS15}
\bibfield{author}{\bibinfo{person}{Timothy~P. Lillicrap}, \bibinfo{person}{Jonathan~J. Hunt}, \bibinfo{person}{Alexander Pritzel}, \bibinfo{person}{Nicolas Heess}, \bibinfo{person}{Tom Erez}, \bibinfo{person}{Yuval Tassa}, \bibinfo{person}{David Silver}, {and} \bibinfo{person}{Daan Wierstra}.} \bibinfo{year}{2016}\natexlab{}.
\newblock \showarticletitle{Continuous control with deep reinforcement learning}. In \bibinfo{booktitle}{\emph{4th International Conference on Learning Representations, {ICLR} 2016, San Juan, Puerto Rico, May 2-4, 2016, Conference Track Proceedings}}.
\newblock


\bibitem[Lin(2004)]%
        {lin2004rouge}
\bibfield{author}{\bibinfo{person}{Chin-Yew Lin}.} \bibinfo{year}{2004}\natexlab{}.
\newblock \showarticletitle{Rouge: A package for automatic evaluation of summaries}. In \bibinfo{booktitle}{\emph{Text summarization branches out}}. \bibinfo{pages}{74--81}.
\newblock


\bibitem[Lin et~al\mbox{.}(2019)]%
        {lin2019pareto}
\bibfield{author}{\bibinfo{person}{Xiao Lin}, \bibinfo{person}{Hongjie Chen}, \bibinfo{person}{Changhua Pei}, \bibinfo{person}{Fei Sun}, \bibinfo{person}{Xuanji Xiao}, \bibinfo{person}{Hanxiao Sun}, \bibinfo{person}{Yongfeng Zhang}, \bibinfo{person}{Wenwu Ou}, {and} \bibinfo{person}{Peng Jiang}.} \bibinfo{year}{2019}\natexlab{}.
\newblock \showarticletitle{A pareto-efficient algorithm for multiple objective optimization in e-commerce recommendation}. In \bibinfo{booktitle}{\emph{Proceedings of the 13th ACM Conference on recommender systems}}. \bibinfo{pages}{20--28}.
\newblock


\bibitem[Luc(2008)]%
        {luc2008pareto}
\bibfield{author}{\bibinfo{person}{Dinh~The Luc}.} \bibinfo{year}{2008}\natexlab{}.
\newblock \showarticletitle{Pareto optimality}.
\newblock \bibinfo{journal}{\emph{Pareto optimality, game theory and equilibria}} (\bibinfo{year}{2008}), \bibinfo{pages}{481--515}.
\newblock


\bibitem[Mnih and Salakhutdinov(2007)]%
        {mnih2007probabilistic}
\bibfield{author}{\bibinfo{person}{Andriy Mnih} {and} \bibinfo{person}{Russ~R Salakhutdinov}.} \bibinfo{year}{2007}\natexlab{}.
\newblock \showarticletitle{Probabilistic matrix factorization}.
\newblock \bibinfo{journal}{\emph{Advances in neural information processing systems}}  \bibinfo{volume}{20} (\bibinfo{year}{2007}).
\newblock


\bibitem[Mnih et~al\mbox{.}(2016)]%
        {mnih2016asynchronous}
\bibfield{author}{\bibinfo{person}{Volodymyr Mnih}, \bibinfo{person}{Adria~Puigdomenech Badia}, \bibinfo{person}{Mehdi Mirza}, \bibinfo{person}{Alex Graves}, \bibinfo{person}{Timothy Lillicrap}, \bibinfo{person}{Tim Harley}, \bibinfo{person}{David Silver}, {and} \bibinfo{person}{Koray Kavukcuoglu}.} \bibinfo{year}{2016}\natexlab{}.
\newblock \showarticletitle{Asynchronous methods for deep reinforcement learning}. In \bibinfo{booktitle}{\emph{International conference on machine learning}}. PMLR, \bibinfo{pages}{1928--1937}.
\newblock


\bibitem[Mnih et~al\mbox{.}(2015)]%
        {mnih2015human}
\bibfield{author}{\bibinfo{person}{Volodymyr Mnih}, \bibinfo{person}{Koray Kavukcuoglu}, \bibinfo{person}{David Silver}, \bibinfo{person}{Andrei~A Rusu}, \bibinfo{person}{Joel Veness}, \bibinfo{person}{Marc~G Bellemare}, \bibinfo{person}{Alex Graves}, \bibinfo{person}{Martin Riedmiller}, \bibinfo{person}{Andreas~K Fidjeland}, \bibinfo{person}{Georg Ostrovski}, {et~al\mbox{.}}} \bibinfo{year}{2015}\natexlab{}.
\newblock \showarticletitle{Human-level control through deep reinforcement learning}.
\newblock \bibinfo{journal}{\emph{nature}} \bibinfo{volume}{518}, \bibinfo{number}{7540} (\bibinfo{year}{2015}), \bibinfo{pages}{529--533}.
\newblock


\bibitem[Ovaisi et~al\mbox{.}(2022)]%
        {ovaisi2022rgrecsys}
\bibfield{author}{\bibinfo{person}{Zohreh Ovaisi}, \bibinfo{person}{Shelby Heinecke}, \bibinfo{person}{Jia Li}, \bibinfo{person}{Yongfeng Zhang}, \bibinfo{person}{Elena Zheleva}, {and} \bibinfo{person}{Caiming Xiong}.} \bibinfo{year}{2022}\natexlab{}.
\newblock \showarticletitle{RGRecSys: A toolkit for robustness evaluation of recommender systems}. In \bibinfo{booktitle}{\emph{Proceedings of the Fifteenth ACM International Conference on Web Search and Data Mining}}. \bibinfo{pages}{1597--1600}.
\newblock


\bibitem[Papineni et~al\mbox{.}(2001)]%
        {Papineni_Roukos_Ward_Zhu_2001}
\bibfield{author}{\bibinfo{person}{Kishore Papineni}, \bibinfo{person}{Salim Roukos}, \bibinfo{person}{Todd Ward}, {and} \bibinfo{person}{Wei-Jing Zhu}.} \bibinfo{year}{2001}\natexlab{}.
\newblock \showarticletitle{BLEU}. In \bibinfo{booktitle}{\emph{Proceedings of the 40th Annual Meeting on Association for Computational Linguistics - ACL ’02}}.
\newblock


\bibitem[Paszke et~al\mbox{.}(2019)]%
        {paszke2019pytorch}
\bibfield{author}{\bibinfo{person}{Adam Paszke}, \bibinfo{person}{Sam Gross}, \bibinfo{person}{Francisco Massa}, \bibinfo{person}{Adam Lerer}, \bibinfo{person}{James Bradbury}, \bibinfo{person}{Gregory Chanan}, \bibinfo{person}{Trevor Killeen}, \bibinfo{person}{Zeming Lin}, \bibinfo{person}{Natalia Gimelshein}, \bibinfo{person}{Luca Antiga}, {et~al\mbox{.}}} \bibinfo{year}{2019}\natexlab{}.
\newblock \showarticletitle{Pytorch: An imperative style, high-performance deep learning library}.
\newblock \bibinfo{journal}{\emph{Advances in neural information processing systems}}  \bibinfo{volume}{32} (\bibinfo{year}{2019}).
\newblock


\bibitem[Peake and Wang(2018)]%
        {peake2018explanation}
\bibfield{author}{\bibinfo{person}{Georgina Peake} {and} \bibinfo{person}{Jun Wang}.} \bibinfo{year}{2018}\natexlab{}.
\newblock \showarticletitle{Explanation mining: Post hoc interpretability of latent factor models for recommendation systems}. In \bibinfo{booktitle}{\emph{Proceedings of the 24th ACM SIGKDD International Conference on Knowledge Discovery \& Data Mining}}. \bibinfo{pages}{2060--2069}.
\newblock


\bibitem[Radford et~al\mbox{.}(2019)]%
        {radford2019language}
\bibfield{author}{\bibinfo{person}{Alec Radford}, \bibinfo{person}{Jeffrey Wu}, \bibinfo{person}{Rewon Child}, \bibinfo{person}{David Luan}, \bibinfo{person}{Dario Amodei}, \bibinfo{person}{Ilya Sutskever}, {et~al\mbox{.}}} \bibinfo{year}{2019}\natexlab{}.
\newblock \showarticletitle{Language models are unsupervised multitask learners}.
\newblock \bibinfo{journal}{\emph{OpenAI blog}} \bibinfo{volume}{1}, \bibinfo{number}{8} (\bibinfo{year}{2019}), \bibinfo{pages}{9}.
\newblock


\bibitem[Rame et~al\mbox{.}(2023)]%
        {rame2023rewarded}
\bibfield{author}{\bibinfo{person}{Alexandre Rame}, \bibinfo{person}{Guillaume Couairon}, \bibinfo{person}{Corentin Dancette}, \bibinfo{person}{Jean-Baptiste Gaya}, \bibinfo{person}{Mustafa Shukor}, \bibinfo{person}{Laure Soulier}, {and} \bibinfo{person}{Matthieu Cord}.} \bibinfo{year}{2023}\natexlab{}.
\newblock \showarticletitle{Rewarded soups: towards pareto-optimal alignment by interpolating weights fine-tuned on diverse rewards}.
\newblock \bibinfo{journal}{\emph{Advances in Neural Information Processing Systems}}  \bibinfo{volume}{36} (\bibinfo{year}{2023}), \bibinfo{pages}{71095--71134}.
\newblock


\bibitem[Ribeiro et~al\mbox{.}(2012)]%
        {ribeiro2012pareto}
\bibfield{author}{\bibinfo{person}{Marco~Tulio Ribeiro}, \bibinfo{person}{Anisio Lacerda}, \bibinfo{person}{Adriano Veloso}, {and} \bibinfo{person}{Nivio Ziviani}.} \bibinfo{year}{2012}\natexlab{}.
\newblock \showarticletitle{Pareto-efficient hybridization for multi-objective recommender systems}. In \bibinfo{booktitle}{\emph{Proceedings of the sixth ACM conference on Recommender systems}}. \bibinfo{pages}{19--26}.
\newblock


\bibitem[Schulman et~al\mbox{.}(2015)]%
        {schulman2015trust}
\bibfield{author}{\bibinfo{person}{John Schulman}, \bibinfo{person}{Sergey Levine}, \bibinfo{person}{Pieter Abbeel}, \bibinfo{person}{Michael Jordan}, {and} \bibinfo{person}{Philipp Moritz}.} \bibinfo{year}{2015}\natexlab{}.
\newblock \showarticletitle{Trust region policy optimization}. In \bibinfo{booktitle}{\emph{International conference on machine learning}}. PMLR, \bibinfo{pages}{1889--1897}.
\newblock


\bibitem[Schulman et~al\mbox{.}(2017)]%
        {schulman2017proximal}
\bibfield{author}{\bibinfo{person}{John Schulman}, \bibinfo{person}{Filip Wolski}, \bibinfo{person}{Prafulla Dhariwal}, \bibinfo{person}{Alec Radford}, {and} \bibinfo{person}{Oleg Klimov}.} \bibinfo{year}{2017}\natexlab{}.
\newblock \showarticletitle{Proximal policy optimization algorithms}.
\newblock \bibinfo{journal}{\emph{arXiv preprint arXiv:1707.06347}} (\bibinfo{year}{2017}).
\newblock


\bibitem[Sener and Koltun(2018)]%
        {sener2018multi}
\bibfield{author}{\bibinfo{person}{Ozan Sener} {and} \bibinfo{person}{Vladlen Koltun}.} \bibinfo{year}{2018}\natexlab{}.
\newblock \showarticletitle{Multi-task learning as multi-objective optimization}.
\newblock \bibinfo{journal}{\emph{Advances in neural information processing systems}}  \bibinfo{volume}{31} (\bibinfo{year}{2018}).
\newblock


\bibitem[Shi et~al\mbox{.}(2020)]%
        {shi2020neural}
\bibfield{author}{\bibinfo{person}{Shaoyun Shi}, \bibinfo{person}{Hanxiong Chen}, \bibinfo{person}{Weizhi Ma}, \bibinfo{person}{Jiaxin Mao}, \bibinfo{person}{Min Zhang}, {and} \bibinfo{person}{Yongfeng Zhang}.} \bibinfo{year}{2020}\natexlab{}.
\newblock \showarticletitle{Neural logic reasoning}. In \bibinfo{booktitle}{\emph{Proceedings of the 29th ACM International Conference on Information \& Knowledge Management}}. \bibinfo{pages}{1365--1374}.
\newblock


\bibitem[Tan et~al\mbox{.}(2021)]%
        {tan2021counterfactual}
\bibfield{author}{\bibinfo{person}{Juntao Tan}, \bibinfo{person}{Shuyuan Xu}, \bibinfo{person}{Yingqiang Ge}, \bibinfo{person}{Yunqi Li}, \bibinfo{person}{Xu Chen}, {and} \bibinfo{person}{Yongfeng Zhang}.} \bibinfo{year}{2021}\natexlab{}.
\newblock \showarticletitle{Counterfactual explainable recommendation}. In \bibinfo{booktitle}{\emph{Proceedings of the 30th ACM International Conference on Information \& Knowledge Management}}. \bibinfo{pages}{1784--1793}.
\newblock


\bibitem[Tang et~al\mbox{.}(2023)]%
        {tang2023fairness}
\bibfield{author}{\bibinfo{person}{Jiakai Tang}, \bibinfo{person}{Shiqi Shen}, \bibinfo{person}{Zhipeng Wang}, \bibinfo{person}{Zhi Gong}, \bibinfo{person}{Jingsen Zhang}, {and} \bibinfo{person}{Xu Chen}.} \bibinfo{year}{2023}\natexlab{}.
\newblock \showarticletitle{When Fairness meets Bias: a Debiased Framework for Fairness aware Top-N Recommendation}. In \bibinfo{booktitle}{\emph{Proceedings of the 17th ACM Conference on Recommender Systems}}. \bibinfo{pages}{200--210}.
\newblock


\bibitem[Touvron et~al\mbox{.}(2023)]%
        {touvron2023llama}
\bibfield{author}{\bibinfo{person}{Hugo Touvron}, \bibinfo{person}{Thibaut Lavril}, \bibinfo{person}{Gautier Izacard}, \bibinfo{person}{Xavier Martinet}, \bibinfo{person}{Marie-Anne Lachaux}, \bibinfo{person}{Timoth{\'e}e Lacroix}, \bibinfo{person}{Baptiste Rozi{\`e}re}, \bibinfo{person}{Naman Goyal}, \bibinfo{person}{Eric Hambro}, \bibinfo{person}{Faisal Azhar}, {et~al\mbox{.}}} \bibinfo{year}{2023}\natexlab{}.
\newblock \showarticletitle{Llama: Open and efficient foundation language models}.
\newblock \bibinfo{journal}{\emph{arXiv preprint arXiv:2302.13971}} (\bibinfo{year}{2023}).
\newblock


\bibitem[Van~Hasselt et~al\mbox{.}(2016)]%
        {van2016deep}
\bibfield{author}{\bibinfo{person}{Hado Van~Hasselt}, \bibinfo{person}{Arthur Guez}, {and} \bibinfo{person}{David Silver}.} \bibinfo{year}{2016}\natexlab{}.
\newblock \showarticletitle{Deep reinforcement learning with double q-learning}. In \bibinfo{booktitle}{\emph{Proceedings of the AAAI conference on artificial intelligence}}, Vol.~\bibinfo{volume}{30}.
\newblock


\bibitem[Vaswani et~al\mbox{.}(2017)]%
        {vaswani2017attention}
\bibfield{author}{\bibinfo{person}{Ashish Vaswani}, \bibinfo{person}{Noam Shazeer}, \bibinfo{person}{Niki Parmar}, \bibinfo{person}{Jakob Uszkoreit}, \bibinfo{person}{Llion Jones}, \bibinfo{person}{Aidan~N Gomez}, \bibinfo{person}{{\L}ukasz Kaiser}, {and} \bibinfo{person}{Illia Polosukhin}.} \bibinfo{year}{2017}\natexlab{}.
\newblock \showarticletitle{Attention is all you need}.
\newblock \bibinfo{journal}{\emph{Advances in neural information processing systems}}  \bibinfo{volume}{30} (\bibinfo{year}{2017}).
\newblock


\bibitem[Wang et~al\mbox{.}(2023a)]%
        {wang2023recagent}
\bibfield{author}{\bibinfo{person}{Lei Wang}, \bibinfo{person}{Jingsen Zhang}, \bibinfo{person}{Xu Chen}, \bibinfo{person}{Yankai Lin}, \bibinfo{person}{Ruihua Song}, \bibinfo{person}{Wayne~Xin Zhao}, {and} \bibinfo{person}{Ji-Rong Wen}.} \bibinfo{year}{2023}\natexlab{a}.
\newblock \showarticletitle{Recagent: A novel simulation paradigm for recommender systems}.
\newblock \bibinfo{journal}{\emph{arXiv preprint arXiv:2306.02552}} (\bibinfo{year}{2023}).
\newblock


\bibitem[Wang et~al\mbox{.}(2023b)]%
        {wang2023user}
\bibfield{author}{\bibinfo{person}{Lei Wang}, \bibinfo{person}{Jingsen Zhang}, \bibinfo{person}{Hao Yang}, \bibinfo{person}{Zhiyuan Chen}, \bibinfo{person}{Jiakai Tang}, \bibinfo{person}{Zeyu Zhang}, \bibinfo{person}{Xu Chen}, \bibinfo{person}{Yankai Lin}, \bibinfo{person}{Ruihua Song}, \bibinfo{person}{Wayne~Xin Zhao}, {et~al\mbox{.}}} \bibinfo{year}{2023}\natexlab{b}.
\newblock \showarticletitle{User behavior simulation with large language model based agents}.
\newblock \bibinfo{journal}{\emph{arXiv preprint arXiv:2306.02552}} (\bibinfo{year}{2023}).
\newblock


\bibitem[Wang et~al\mbox{.}(2018)]%
        {wang2018explainable}
\bibfield{author}{\bibinfo{person}{Nan Wang}, \bibinfo{person}{Hongning Wang}, \bibinfo{person}{Yiling Jia}, {and} \bibinfo{person}{Yue Yin}.} \bibinfo{year}{2018}\natexlab{}.
\newblock \showarticletitle{Explainable recommendation via multi-task learning in opinionated text data}. In \bibinfo{booktitle}{\emph{The 41st international ACM SIGIR conference on research \& development in information retrieval}}. \bibinfo{pages}{165--174}.
\newblock


\bibitem[Wang et~al\mbox{.}(2024)]%
        {wang2024reinforced}
\bibfield{author}{\bibinfo{person}{Xiangmeng Wang}, \bibinfo{person}{Qian Li}, \bibinfo{person}{Dianer Yu}, \bibinfo{person}{Qing Li}, {and} \bibinfo{person}{Guandong Xu}.} \bibinfo{year}{2024}\natexlab{}.
\newblock \showarticletitle{Reinforced path reasoning for counterfactual explainable recommendation}.
\newblock \bibinfo{journal}{\emph{IEEE Transactions on Knowledge and Data Engineering}} (\bibinfo{year}{2024}).
\newblock


\bibitem[Wang et~al\mbox{.}(2019)]%
        {wang2019explainable}
\bibfield{author}{\bibinfo{person}{Xiang Wang}, \bibinfo{person}{Dingxian Wang}, \bibinfo{person}{Canran Xu}, \bibinfo{person}{Xiangnan He}, \bibinfo{person}{Yixin Cao}, {and} \bibinfo{person}{Tat-Seng Chua}.} \bibinfo{year}{2019}\natexlab{}.
\newblock \showarticletitle{Explainable reasoning over knowledge graphs for recommendation}. In \bibinfo{booktitle}{\emph{Proceedings of the AAAI conference on artificial intelligence}}, Vol.~\bibinfo{volume}{33}. \bibinfo{pages}{5329--5336}.
\newblock


\bibitem[Watkins and Dayan(1992)]%
        {watkins1992q}
\bibfield{author}{\bibinfo{person}{Christopher~JCH Watkins} {and} \bibinfo{person}{Peter Dayan}.} \bibinfo{year}{1992}\natexlab{}.
\newblock \showarticletitle{Q-learning}.
\newblock \bibinfo{journal}{\emph{Machine learning}}  \bibinfo{volume}{8} (\bibinfo{year}{1992}), \bibinfo{pages}{279--292}.
\newblock


\bibitem[Wei et~al\mbox{.}(2022a)]%
        {DBLP:journals/tmlr/WeiTBRZBYBZMCHVLDF22}
\bibfield{author}{\bibinfo{person}{Jason Wei}, \bibinfo{person}{Yi Tay}, \bibinfo{person}{Rishi Bommasani}, \bibinfo{person}{Colin Raffel}, \bibinfo{person}{Barret Zoph}, \bibinfo{person}{Sebastian Borgeaud}, \bibinfo{person}{Dani Yogatama}, \bibinfo{person}{Maarten Bosma}, \bibinfo{person}{Denny Zhou}, \bibinfo{person}{Donald Metzler}, \bibinfo{person}{Ed~H. Chi}, \bibinfo{person}{Tatsunori Hashimoto}, \bibinfo{person}{Oriol Vinyals}, \bibinfo{person}{Percy Liang}, \bibinfo{person}{Jeff Dean}, {and} \bibinfo{person}{William Fedus}.} \bibinfo{year}{2022}\natexlab{a}.
\newblock \showarticletitle{Emergent Abilities of Large Language Models}.
\newblock \bibinfo{journal}{\emph{Trans. Mach. Learn. Res.}}  \bibinfo{volume}{2022} (\bibinfo{year}{2022}).
\newblock


\bibitem[Wei et~al\mbox{.}(2022b)]%
        {wei2022chain}
\bibfield{author}{\bibinfo{person}{Jason Wei}, \bibinfo{person}{Xuezhi Wang}, \bibinfo{person}{Dale Schuurmans}, \bibinfo{person}{Maarten Bosma}, \bibinfo{person}{Fei Xia}, \bibinfo{person}{Ed Chi}, \bibinfo{person}{Quoc~V Le}, \bibinfo{person}{Denny Zhou}, {et~al\mbox{.}}} \bibinfo{year}{2022}\natexlab{b}.
\newblock \showarticletitle{Chain-of-thought prompting elicits reasoning in large language models}.
\newblock \bibinfo{journal}{\emph{Advances in neural information processing systems}}  \bibinfo{volume}{35} (\bibinfo{year}{2022}), \bibinfo{pages}{24824--24837}.
\newblock


\bibitem[Xian et~al\mbox{.}(2019)]%
        {xian2019reinforcement}
\bibfield{author}{\bibinfo{person}{Yikun Xian}, \bibinfo{person}{Zuohui Fu}, \bibinfo{person}{Shan Muthukrishnan}, \bibinfo{person}{Gerard De~Melo}, {and} \bibinfo{person}{Yongfeng Zhang}.} \bibinfo{year}{2019}\natexlab{}.
\newblock \showarticletitle{Reinforcement knowledge graph reasoning for explainable recommendation}. In \bibinfo{booktitle}{\emph{Proceedings of the 42nd international ACM SIGIR conference on research and development in information retrieval}}. \bibinfo{pages}{285--294}.
\newblock


\bibitem[Xian et~al\mbox{.}(2020)]%
        {xian2020cafe}
\bibfield{author}{\bibinfo{person}{Yikun Xian}, \bibinfo{person}{Zuohui Fu}, \bibinfo{person}{Handong Zhao}, \bibinfo{person}{Yingqiang Ge}, \bibinfo{person}{Xu Chen}, \bibinfo{person}{Qiaoying Huang}, \bibinfo{person}{Shijie Geng}, \bibinfo{person}{Zhou Qin}, \bibinfo{person}{Gerard De~Melo}, \bibinfo{person}{Shan Muthukrishnan}, {et~al\mbox{.}}} \bibinfo{year}{2020}\natexlab{}.
\newblock \showarticletitle{CAFE: Coarse-to-fine neural symbolic reasoning for explainable recommendation}. In \bibinfo{booktitle}{\emph{Proceedings of the 29th ACM International Conference on Information \& Knowledge Management}}. \bibinfo{pages}{1645--1654}.
\newblock


\bibitem[Xie et~al\mbox{.}(2021a)]%
        {xie2021personalized}
\bibfield{author}{\bibinfo{person}{Ruobing Xie}, \bibinfo{person}{Yanlei Liu}, \bibinfo{person}{Shaoliang Zhang}, \bibinfo{person}{Rui Wang}, \bibinfo{person}{Feng Xia}, {and} \bibinfo{person}{Leyu Lin}.} \bibinfo{year}{2021}\natexlab{a}.
\newblock \showarticletitle{Personalized approximate pareto-efficient recommendation}. In \bibinfo{booktitle}{\emph{Proceedings of the Web Conference 2021}}. \bibinfo{pages}{3839--3849}.
\newblock


\bibitem[Xie et~al\mbox{.}(2021b)]%
        {10.1145/3442381.3450039}
\bibfield{author}{\bibinfo{person}{Ruobing Xie}, \bibinfo{person}{Yanlei Liu}, \bibinfo{person}{Shaoliang Zhang}, \bibinfo{person}{Rui Wang}, \bibinfo{person}{Feng Xia}, {and} \bibinfo{person}{Leyu Lin}.} \bibinfo{year}{2021}\natexlab{b}.
\newblock \showarticletitle{Personalized Approximate Pareto-Efficient Recommendation}. In \bibinfo{booktitle}{\emph{Proceedings of the Web Conference 2021}} \emph{(\bibinfo{series}{WWW '21})}. \bibinfo{publisher}{Association for Computing Machinery}, \bibinfo{address}{New York, NY, USA}, \bibinfo{pages}{3839–3849}.
\newblock
\showISBNx{9781450383127}


\bibitem[Xie et~al\mbox{.}(2023)]%
        {xie2023factual}
\bibfield{author}{\bibinfo{person}{Zhouhang Xie}, \bibinfo{person}{Sameer Singh}, \bibinfo{person}{Julian McAuley}, {and} \bibinfo{person}{Bodhisattwa~Prasad Majumder}.} \bibinfo{year}{2023}\natexlab{}.
\newblock \showarticletitle{Factual and informative review generation for explainable recommendation}. In \bibinfo{booktitle}{\emph{Proceedings of the AAAI Conference on Artificial Intelligence}}, Vol.~\bibinfo{volume}{37}. \bibinfo{pages}{13816--13824}.
\newblock


\bibitem[Yang et~al\mbox{.}(2024)]%
        {yang2024fine}
\bibfield{author}{\bibinfo{person}{Mengyuan Yang}, \bibinfo{person}{Mengying Zhu}, \bibinfo{person}{Yan Wang}, \bibinfo{person}{Linxun Chen}, \bibinfo{person}{Yilei Zhao}, \bibinfo{person}{Xiuyuan Wang}, \bibinfo{person}{Bing Han}, \bibinfo{person}{Xiaolin Zheng}, {and} \bibinfo{person}{Jianwei Yin}.} \bibinfo{year}{2024}\natexlab{}.
\newblock \showarticletitle{Fine-Tuning Large Language Model Based Explainable Recommendation with Explainable Quality Reward}. In \bibinfo{booktitle}{\emph{Proceedings of the AAAI Conference on Artificial Intelligence}}, Vol.~\bibinfo{volume}{38}. \bibinfo{pages}{9250--9259}.
\newblock


\bibitem[Yoon et~al\mbox{.}(2024)]%
        {yoon2024evaluating}
\bibfield{author}{\bibinfo{person}{Se-eun Yoon}, \bibinfo{person}{Zhankui He}, \bibinfo{person}{Jessica~Maria Echterhoff}, {and} \bibinfo{person}{Julian McAuley}.} \bibinfo{year}{2024}\natexlab{}.
\newblock \showarticletitle{Evaluating large language models as generative user simulators for conversational recommendation}.
\newblock \bibinfo{journal}{\emph{arXiv preprint arXiv:2403.09738}} (\bibinfo{year}{2024}).
\newblock


\bibitem[Zar(2005)]%
        {zar2005spearman}
\bibfield{author}{\bibinfo{person}{Jerrold~H Zar}.} \bibinfo{year}{2005}\natexlab{}.
\newblock \showarticletitle{Spearman rank correlation}.
\newblock \bibinfo{journal}{\emph{Encyclopedia of biostatistics}}  \bibinfo{volume}{7} (\bibinfo{year}{2005}).
\newblock


\bibitem[Zeng et~al\mbox{.}(2023)]%
        {DBLP:conf/iclr/ZengLDWL0YXZXTM23}
\bibfield{author}{\bibinfo{person}{Aohan Zeng}, \bibinfo{person}{Xiao Liu}, \bibinfo{person}{Zhengxiao Du}, \bibinfo{person}{Zihan Wang}, \bibinfo{person}{Hanyu Lai}, \bibinfo{person}{Ming Ding}, \bibinfo{person}{Zhuoyi Yang}, \bibinfo{person}{Yifan Xu}, \bibinfo{person}{Wendi Zheng}, \bibinfo{person}{Xiao Xia}, \bibinfo{person}{Weng~Lam Tam}, \bibinfo{person}{Zixuan Ma}, \bibinfo{person}{Yufei Xue}, \bibinfo{person}{Jidong Zhai}, \bibinfo{person}{Wenguang Chen}, \bibinfo{person}{Zhiyuan Liu}, \bibinfo{person}{Peng Zhang}, \bibinfo{person}{Yuxiao Dong}, {and} \bibinfo{person}{Jie Tang}.} \bibinfo{year}{2023}\natexlab{}.
\newblock \showarticletitle{{GLM-130B:} An Open Bilingual Pre-trained Model}. In \bibinfo{booktitle}{\emph{The Eleventh International Conference on Learning Representations, {ICLR} 2023, Kigali, Rwanda, May 1-5, 2023}}.
\newblock


\bibitem[Zhang et~al\mbox{.}(2023)]%
        {zhang2023recommendation}
\bibfield{author}{\bibinfo{person}{Jingsen Zhang}, \bibinfo{person}{Xu Chen}, \bibinfo{person}{Jiakai Tang}, \bibinfo{person}{Weiqi Shao}, \bibinfo{person}{Quanyu Dai}, \bibinfo{person}{Zhenhua Dong}, {and} \bibinfo{person}{Rui Zhang}.} \bibinfo{year}{2023}\natexlab{}.
\newblock \showarticletitle{Recommendation with causality enhanced natural language explanations}. In \bibinfo{booktitle}{\emph{Proceedings of the ACM Web Conference 2023}}. \bibinfo{pages}{876--886}.
\newblock


\bibitem[Zhang et~al\mbox{.}(2024)]%
        {10.1145/3589334.3645537}
\bibfield{author}{\bibinfo{person}{Junjie Zhang}, \bibinfo{person}{Yupeng Hou}, \bibinfo{person}{Ruobing Xie}, \bibinfo{person}{Wenqi Sun}, \bibinfo{person}{Julian McAuley}, \bibinfo{person}{Wayne~Xin Zhao}, \bibinfo{person}{Leyu Lin}, {and} \bibinfo{person}{Ji-Rong Wen}.} \bibinfo{year}{2024}\natexlab{}.
\newblock \showarticletitle{AgentCF: Collaborative Learning with Autonomous Language Agents for Recommender Systems}. In \bibinfo{booktitle}{\emph{Proceedings of the ACM on Web Conference 2024}} \emph{(\bibinfo{series}{WWW '24})}. \bibinfo{publisher}{Association for Computing Machinery}, \bibinfo{address}{New York, NY, USA}, \bibinfo{pages}{3679–3689}.
\newblock
\showISBNx{9798400701719}


\bibitem[Zhang* et~al\mbox{.}(2020)]%
        {bert-score}
\bibfield{author}{\bibinfo{person}{Tianyi Zhang*}, \bibinfo{person}{Varsha Kishore*}, \bibinfo{person}{Felix Wu*}, \bibinfo{person}{Kilian~Q. Weinberger}, {and} \bibinfo{person}{Yoav Artzi}.} \bibinfo{year}{2020}\natexlab{}.
\newblock \showarticletitle{BERTScore: Evaluating Text Generation with BERT}. In \bibinfo{booktitle}{\emph{International Conference on Learning Representations}}.
\newblock


\bibitem[Zhang and Chen(2020)]%
        {Zhang_2020}
\bibfield{author}{\bibinfo{person}{Yongfeng Zhang} {and} \bibinfo{person}{Xu Chen}.} \bibinfo{year}{2020}\natexlab{}.
\newblock \showarticletitle{Explainable Recommendation: A Survey and New Perspectives}.
\newblock \bibinfo{journal}{\emph{Foundations and Trends® in Information Retrieval}} \bibinfo{volume}{14}, \bibinfo{number}{1} (\bibinfo{year}{2020}), \bibinfo{pages}{1–101}.
\newblock
\showISSN{1554-0677}
\urldef\tempurl%
\url{https://doi.org/10.1561/1500000066}
\showDOI{\tempurl}


\bibitem[Zhang et~al\mbox{.}(2014)]%
        {zhang2014explicit}
\bibfield{author}{\bibinfo{person}{Yongfeng Zhang}, \bibinfo{person}{Guokun Lai}, \bibinfo{person}{Min Zhang}, \bibinfo{person}{Yi Zhang}, \bibinfo{person}{Yiqun Liu}, {and} \bibinfo{person}{Shaoping Ma}.} \bibinfo{year}{2014}\natexlab{}.
\newblock \showarticletitle{Explicit factor models for explainable recommendation based on phrase-level sentiment analysis}. In \bibinfo{booktitle}{\emph{Proceedings of the 37th international ACM SIGIR conference on Research \& development in information retrieval}}. \bibinfo{pages}{83--92}.
\newblock


\bibitem[Zhao et~al\mbox{.}(2021)]%
        {10.1145/3459637.3482016}
\bibfield{author}{\bibinfo{person}{Wayne~Xin Zhao}, \bibinfo{person}{Shanlei Mu}, \bibinfo{person}{Yupeng Hou}, \bibinfo{person}{Zihan Lin}, \bibinfo{person}{Yushuo Chen}, \bibinfo{person}{Xingyu Pan}, \bibinfo{person}{Kaiyuan Li}, \bibinfo{person}{Yujie Lu}, \bibinfo{person}{Hui Wang}, \bibinfo{person}{Changxin Tian}, \bibinfo{person}{Yingqian Min}, \bibinfo{person}{Zhichao Feng}, \bibinfo{person}{Xinyan Fan}, \bibinfo{person}{Xu Chen}, \bibinfo{person}{Pengfei Wang}, \bibinfo{person}{Wendi Ji}, \bibinfo{person}{Yaliang Li}, \bibinfo{person}{Xiaoling Wang}, {and} \bibinfo{person}{Ji-Rong Wen}.} \bibinfo{year}{2021}\natexlab{}.
\newblock \showarticletitle{RecBole: Towards a Unified, Comprehensive and Efficient Framework for Recommendation Algorithms}. In \bibinfo{booktitle}{\emph{Proceedings of the 30th ACM International Conference on Information \& Knowledge Management}} \emph{(\bibinfo{series}{CIKM '21})}. \bibinfo{publisher}{Association for Computing Machinery}, \bibinfo{address}{New York, NY, USA}, \bibinfo{pages}{4653–4664}.
\newblock
\showISBNx{9781450384469}


\bibitem[Zhao et~al\mbox{.}(2023)]%
        {zhao2023survey}
\bibfield{author}{\bibinfo{person}{Wayne~Xin Zhao}, \bibinfo{person}{Kun Zhou}, \bibinfo{person}{Junyi Li}, \bibinfo{person}{Tianyi Tang}, \bibinfo{person}{Xiaolei Wang}, \bibinfo{person}{Yupeng Hou}, \bibinfo{person}{Yingqian Min}, \bibinfo{person}{Beichen Zhang}, \bibinfo{person}{Junjie Zhang}, \bibinfo{person}{Zican Dong}, {et~al\mbox{.}}} \bibinfo{year}{2023}\natexlab{}.
\newblock \showarticletitle{A survey of large language models}.
\newblock \bibinfo{journal}{\emph{arXiv preprint arXiv:2303.18223}} (\bibinfo{year}{2023}).
\newblock


\bibitem[Zheng et~al\mbox{.}(2018)]%
        {zheng2018dags}
\bibfield{author}{\bibinfo{person}{Xun Zheng}, \bibinfo{person}{Bryon Aragam}, \bibinfo{person}{Pradeep~K Ravikumar}, {and} \bibinfo{person}{Eric~P Xing}.} \bibinfo{year}{2018}\natexlab{}.
\newblock \showarticletitle{Dags with no tears: Continuous optimization for structure learning}.
\newblock \bibinfo{journal}{\emph{Advances in neural information processing systems}}  \bibinfo{volume}{31} (\bibinfo{year}{2018}).
\newblock


\bibitem[Zhou et~al\mbox{.}(2023)]%
        {DBLP:conf/iclr/ZhouSHWS0SCBLC23}
\bibfield{author}{\bibinfo{person}{Denny Zhou}, \bibinfo{person}{Nathanael Sch{\"{a}}rli}, \bibinfo{person}{Le Hou}, \bibinfo{person}{Jason Wei}, \bibinfo{person}{Nathan Scales}, \bibinfo{person}{Xuezhi Wang}, \bibinfo{person}{Dale Schuurmans}, \bibinfo{person}{Claire Cui}, \bibinfo{person}{Olivier Bousquet}, \bibinfo{person}{Quoc~V. Le}, {and} \bibinfo{person}{Ed~H. Chi}.} \bibinfo{year}{2023}\natexlab{}.
\newblock \showarticletitle{Least-to-Most Prompting Enables Complex Reasoning in Large Language Models}. In \bibinfo{booktitle}{\emph{The Eleventh International Conference on Learning Representations, {ICLR} 2023, Kigali, Rwanda, May 1-5, 2023}}.
\newblock


\bibitem[Zhou et~al\mbox{.}(2015)]%
        {zhou2015svd}
\bibfield{author}{\bibinfo{person}{Xun Zhou}, \bibinfo{person}{Jing He}, \bibinfo{person}{Guangyan Huang}, {and} \bibinfo{person}{Yanchun Zhang}.} \bibinfo{year}{2015}\natexlab{}.
\newblock \showarticletitle{SVD-based incremental approaches for recommender systems}.
\newblock \bibinfo{journal}{\emph{J. Comput. System Sci.}} \bibinfo{volume}{81}, \bibinfo{number}{4} (\bibinfo{year}{2015}), \bibinfo{pages}{717--733}.
\newblock


\bibitem[Zhu et~al\mbox{.}(2021)]%
        {zhu-etal-2021-faithfully}
\bibfield{author}{\bibinfo{person}{Yaxin Zhu}, \bibinfo{person}{Yikun Xian}, \bibinfo{person}{Zuohui Fu}, \bibinfo{person}{Gerard de Melo}, {and} \bibinfo{person}{Yongfeng Zhang}.} \bibinfo{year}{2021}\natexlab{}.
\newblock \showarticletitle{Faithfully Explainable Recommendation via Neural Logic Reasoning}. In \bibinfo{booktitle}{\emph{Proceedings of the 2021 Conference of the North American Chapter of the Association for Computational Linguistics: Human Language Technologies}}. \bibinfo{publisher}{Association for Computational Linguistics}.
\newblock


\end{thebibliography}

\end{document}